\def\year{2021}\relax
\documentclass[letterpaper]{article} % DO NOT CHANGE THIS
\usepackage{aaai21}  % DO NOT CHANGE THIS
\usepackage{times}  % DO NOT CHANGE THIS
\usepackage{helvet} % DO NOT CHANGE THIS
\usepackage{courier}  % DO NOT CHANGE THIS
\usepackage[hyphens]{url}  % DO NOT CHANGE THIS
\usepackage{graphicx} % DO NOT CHANGE THIS
\usepackage{natbib}  % DO NOT CHANGE THIS AND DO NOT ADD ANY OPTIONS TO IT
\usepackage{caption} % DO NOT CHANGE THIS AND DO NOT ADD ANY OPTIONS TO IT
\usepackage{algorithm}
\usepackage{algorithmic}

\usepackage{amsmath}
\usepackage{amssymb}
\usepackage{bm}
\usepackage{graphicx}
\usepackage{epsfig}
\usepackage{epstopdf}
\usepackage{subfigure}
\usepackage{array}
\usepackage{multirow}
\usepackage{booktabs}
\usepackage{eucal}
\usepackage{makecell}%bold the line
\usepackage{fnbreak}
\usepackage{amsthm}
\newtheorem{thm}{Theorem}%[section]
\newtheorem{lem}[thm]{Lemma}
\newtheorem{coro}{Corollary}[thm]
\newtheorem{prop}[thm]{Proposition}
\newtheorem{defn}[thm]{Definition}%[section]
\usepackage{tikz} % Package for drawing

\DeclareMathOperator{\Tr}{Tr}
\usepackage[braket,qm]{qcircuit}
\usepackage{qcircuit}

\newcommand{\lr}[1]{\left( #1 \right)}
\newcommand{\Lr}[1]{\left[ #1 \right]}
\newcommand{\LR}[1]{\left\{ #1 \right\}}

\newcommand{\nc}{\newcommand}
\nc{\rank}{\operatorname{Rank}}
\nc{\RR}{{{\mathbb R}}}
\nc{\CC}{{{\mathbb C}}}
\nc{\cA}{{\mathcal{A}}}
\nc{\cB}{{\mathcal{B}}}
\nc{\proj}[1]{| #1\rangle\!\langle #1 |}

\usepackage[switch]{lineno}

\title{VSQL: Variational Shadow Quantum Learning for Classification}
\author {
      Guangxi Li,\textsuperscript{\rm 1,2}
       Zhixin Song,\textsuperscript{\rm 1}
        Xin Wang\textsuperscript{\rm 1}\thanks{Corresponding author.} \\
}
\affiliations {
  \textsuperscript{\rm 1}Institute for Quantum Computing, Baidu Research, Beijing 100193, China \\
   \textsuperscript{\rm 2}Centre for Quantum Software and Information, University of Technology Sydney, NSW 2007, Australia \\
   guangxi.li@student.uts.edu.au, zhixinsong0524@gmail.com, wangxin73@baidu.com
}
\begin{document}

\maketitle

% \linenumbers

\begin{abstract}

Classification of quantum data is essential for quantum machine learning and near-term quantum technologies. In this paper, we propose a new hybrid quantum-classical framework for supervised quantum learning, which we call Variational Shadow Quantum Learning (VSQL). Our method in particular utilizes the classical shadows of quantum data, which fundamentally represent the side information of quantum data with respect to certain physical observables. Specifically, we first use variational shadow quantum circuits to extract classical features in a convolution way and then utilize a fully-connected neural network to complete the classification task. We show that this method could sharply reduce the number of parameters and thus better facilitate quantum circuit training. Simultaneously, less noise will be introduced since fewer quantum gates are employed in such shadow circuits. Moreover, we show that the Barren Plateau issue, a significant gradient vanishing problem in quantum machine learning, could be avoided in VSQL. Finally, we demonstrate the efficiency of VSQL in quantum classification via  numerical experiments on the classification of quantum states and the recognition of multi-labeled handwritten digits. In particular, our VSQL approach outperforms existing variational quantum classifiers in the test accuracy in the binary case of handwritten digit recognition and notably requires much fewer parameters.

\end{abstract}

\section{Introduction}
\label{sec:intro_VSQL}
Quantum computers are expected to have significant applications in solving challenging problems in information processing. Inspired by the powerful capacity of classical supervised learning and its growing community~\cite{LeCun2015,Goodfellow-et-al-2016}, it is natural to develop their quantum counterparts and explore the emerging field of quantum machine learning (QML) \cite{Biamonte2017,Schuld2015a,Ciliberto2018,Lloyd2013}. Among many topics in this area, classification is one of the most important tasks,  e.g., distinguishing quantum states \cite{Bae2015,Chen2018,Patterson2019} or recognizing classical data \cite{Havlicek2019,Benedetti2019,Schuld2018}. Classification is usually described as a decision making process with discrete variable where the processing unit is provided with a labeled training set $\mathcal{D}^{(train)} = \{(\rho^{(m)}, y^{(m)})\}$ in order to find the convoluted mapping $\mathcal{F}$ between each set element $\rho^{(m)}$  and its corresponding label $y^{(m)}$. Once the training process is complete, we would expect the classifier $\mathcal{F}$ not only learns the map $\mathcal{F}(\rho^{(m)}) = y^{(m)}$ precisely, but also generalizes its capacity of discrimination to discover some hidden features shared with similar test data $\mathcal{F}(\rho^{(new)}) = y^{(new)}$ (i.e. recognize an unseen cat as a cat). This ability of generalization is valuable to all classification tasks, and hence it is frequently used to benchmark the performance of a classifier.

In classical machine learning, various approaches have been proposed to implement classification tasks, including perceptron-based algorithms, support vector machines, and the most prevalent neural network (NN) framework \cite{LeCun2015}. With the quantum computing community growing in the NISQ era \cite{Preskill2018}, similar ideas have been developed respectively, including the quantum perceptron model \cite{kapoor2016quantum}, kernel-based method \cite{li2019sublinear}, and the quantum neural network (QNN) framework \cite{Havlicek2019,Mitarai2018,Schuld2018,Grant2018,Farhi2018,Schuld2019,Schuld2018a,Bhatia2019}.  

This paper focuses on the QNN-based algorithms, also referred to as Variational Quantum Algorithms (VQA) or hybrid quantum-classical algorithms since they are regarded as well-suited for execution on NISQ devices by combining quantum computers with classical computers. The main idea of VQA is employing parametrized quantum circuits (as a unitary neural network architecture) to search the Hilbert space and combining classical optimization methods like gradient descent (GD) to find the best parameters \cite{Peruzzo2014,Kandala2017,Farhi2014,Farhi2016}. VQA has been applied to many topics such as quantum eigensolver~\cite{Peruzzo2014}, quantum simulation~\cite{Yuan2018a}, quantum state distance estimation~\cite{Cerezo2020a,Chen2020a} and quantum matrix decomposition~\cite{Wang2020d}.
So far, most proposals for variational quantum classification process information in the global sense such that the quantum circuit always acts on the whole Hilbert space fulfilling high-dimensional transformation.
And the classical feature/information extracted from the quantum system is achieved through measurement. This formulation faces two potential challenges. There could exist more efficient architectures to achieve the same task performance but significantly reducing the quantum resource required (number of quantum gates) by limiting the operating scope to a few selected qubits. 
The other challenge is the notorious Barren Plateau problem \cite{McClean2018}. As the problem size increases, it will exhibit exponentially vanishing gradients, making the optimization landscape flat and hence untrainable using gradient-based optimization methods.

To overcome the above challenges, we explore a significantly different hybrid architecture based on classical shadows. In comparison, our method extracts only "local" features from the subspace, which we call \textit{shadow features}. 
We introduce shadow features for quantum classification due to the inspiration from the recent studies on shadows \cite{Aaronson2018,Huang2020}. Specifically, important quantum properties such as quantum fidelities and entanglement entropies can be predicted using classical shadows rather than possessing full information of the quantum states~\cite{Huang2020}. 
This provides us the intuition that classical shadows may also be applied in quantum classification.

\textbf{Contributions.}
Our work is, as far as we know, the first attempt to combine the concept of shadow with the mindset of learning. We propose a Variational Shadow Quantum Learning (VSQL) framework that could be adapted to many near-term quantum applications.  In particular, we apply this framework to develop quantum classifiers for near-term quantum devices. Firstly, we employ the parameterized shadow quantum circuits $U(\bm{\theta})$ (denoted as \textit{shadow circuits}) acting on selected
% $n_{qsc}$
local
qubit subspace rather than the whole 
% $n$-
qubit Hilbert space, which considers the operating scope efficiency and the connectivity limit on quantum hardware. Secondly, the shadow features 
% $\{o_i\}$
of the input data (encoded as quantum states $\rho^{(m)}$ with labels $y^{(m)}$) will be computed via measuring the Pauli $X\otimes X \cdots \otimes X$ observables on the quantum devices. The final step is to utilize a classical Fully-Connected  Neural Network (FCNN) to post-process these shadow features, and we could then decide the label prediction $\hat{y}^{(m)}$ through an activation function. We refer to~Fig.~\ref{fig:scheme_vsqc} for a detailed sketch of our method for binary classification. 

The contributions of our work are multi-folded. 
First, we introduce a hybrid quantum-classical framework that can be easily implemented on quantum devices with topological connectivity limits, since our VSQL mainly considers locally-operated quantum circuits. Second, we show that VSQL involves significantly fewer number of parameters (independent of the problem size) than existing variational quantum classifiers \cite{Schuld2018,Farhi2018}. Notably, we prove that VSQL could naturally avoid the Barren Plateau issue \cite{McClean2018} (gradients vanishing issue in QML) by limiting the operating scope. 
Finally, we demonstrate real-world applications of VSQL to do quantum state classification and handwritten digits recognition. We in particular show that VSQL outperforms existing variational quantum classifiers in the test accuracy while requiring much fewer parameters.

\section{Method}
\label{sec:method_VSQL}

\subsection{Preliminaries and notations}
\label{sec:quantum_pre}

Here, we briefly introduce the basic concepts of quantum computation that are necessary for this paper. Interested readers are recommended to the celebrated textbook by Nielsen and Chuang \cite{Nielsen2002}. Information in the quantum computing field is represented by $n$-qubit quantum states over space $\CC^{2^n\times 2^n}$, which could be mathematically described by positive semi-definite matrices $\rho\succeq 0$ with property $\Tr(\rho) =1$. Following this density matrix formulation, a quantum state is pure if $\rank(\rho)=1$; otherwise, it is mixed. For a pure state $\rho$, it can be represented by a unit vector in the sense that $\rho = \proj{\psi}$, where the \textit{ket} notation $\ket \psi \in \CC^d$ denotes a column vector and \textit{bra} notation $\bra\psi = \ket \psi^{\dagger}$ with  $\dagger$ denoting conjugate transpose. In general, we would also use $\ket{\psi}$ to denote a pure state for simplicity. A mixed state could be represented as $\rho=\sum_iq_i\op{\psi_i}{\psi_i}$, where the coefficients $q_i\ge 0$ records the probability for quantum system to be in each corresponding pure state $\op{\psi_i}{\psi_i}$ and hence $\sum_iq_i=1$.

In this paper, the unitary matrices
\begin{small}
\begin{align}
    X:= \begin{bmatrix} 0 & 1\\ 1 & 0\end{bmatrix},
    Y:= \begin{bmatrix} 0 & -i\\ i & 0\end{bmatrix},
    Z:= \begin{bmatrix} 1 & 0\\ 0 & -1\end{bmatrix}
\end{align}
\end{small}
refer specifically to the Pauli matrices. Their corresponding rotation gates are denoted by $R_P(\theta):=  \text{e}^{-i\theta P/2} = \cos\frac{\theta}{2}\mathbb I -i\sin\frac{\theta}{2}P$, where $\theta\in [0,2\pi]$ is the rotation angle and $P\in\{X,Y,Z\}$.
The evolution of a quantum state $\rho$ could be mathematically described by employing a quantum circuit (or QNN) $\rho^\prime = U \rho U^\dag$, where the unitary $U$ is usually parameterized by a series of rotation gate angles $\bm{\theta}$ and basic two-qubit gates. Quantum measurement is usually introduced at the end of algorithms to extract classical information from the quantum states. One hardware-friendly way to fulfill this is through the Pauli expectations $\langle P \rangle=\Tr(P\rho^\prime)$, where $P \in\{X,Y,Z\}^{\otimes n}$
is a single Pauli matrix or Pauli product matrix, e.g., $P=X\otimes X\cdots \otimes X$.

\subsection{Sketch of our method}
\label{sec:sketch_VSQL}
We now present the sketch of VSQL for binary classification. Our goal is to find the optimal parameters $\bm{\theta}^*$ in the shadow circuits $U(\bm{\theta})$ and the best weights $\{\bm{w}^*, b^*\}$ in the FCNN such that the algorithm could correctly predict the label of an unknown input quantum state.
Like most classifiers, VSQL consists of two separate processes, training and inference. During the training process illustrated in Fig. \ref{fig:scheme_vsqc}, we are given the training data set encoded in $n$-qubit quantum state $\mathcal{D}^{(train)}:=\{(\rho_{in}^{(m)}, y^{(m)})\}_{m=1}^{N_{train}}$, where $y^{(m)}\in\{0,1\}$ denotes the binary label for the $m$-th input density matrix $\rho_{in}^{(m)}$. Then, the $n_{qsc}$-local shadow circuit acts on the first $n_{qsc}$ qubits and the Pauli-$(X\otimes \cdots \otimes X)$ expectation value is estimated, recorded as \textit{shadow feature} $o_1$.
Next, the identical shadow circuit is implemented on the subspace spanned from the $2^{nd}$ up to the $(2+n_{qsc}-1)^{th}$ qubit to extract the second shadow feature $o_2$. As the shadow circuit sliding down, we obtain $n-n_{qsc}+1$ shadow features in total. 
%Please note that 
This convolution-like way of sliding through the qubit positions can be adjusted according to the hardware connectivity.
We also note that there is only one shadow circuit here. However, the \textbf{n}umber of \textbf{s}hadow circuits ($n_s$) could be increased appropriately to accommodate the difficulty of classification tasks, with $n_s(n-n_{qsc}+1)$ shadow features.
Sequentially, we feed these local features $\{o_i\}$ into a classical FCNN, which means they are weighted summed with weights $\bm{w}\in \mathbb R^{n-n_{qsc}+1}$, bias $b \in \mathbb R$ and mapped into the range $\hat{y}^{(m)} \in [0,1]$ via the sigmoid activation function $\sigma(z) = \lr{1+\text{e}^{-z}}^{-1}$. Repeat the same procedure for each input data and compute the accumulated loss $\mathcal{L}(\bm{\theta},\bm{w},b;\mathcal{D}^{(train)})$
between the predicted value $\hat{y}^{(m)}$ and its true label $y^{(m)}$. Finally, VSQL utilizes gradient-based optimizer to update the shadow circuit parameters $\bm{\theta}$ and 
 the neural network parameters $\bm{w}, b$ and
 gradually minimizing the loss function. Repeat these steps until the loss is converged with tolerance $\Delta \mathcal{L} \leq \varepsilon$ or other stopping criteria is satisfied. Check Algorithm \ref{alg:VSQL_binary_train} for details.

\begin{figure}[t]
	\centering
		\includegraphics[width=0.47\textwidth]{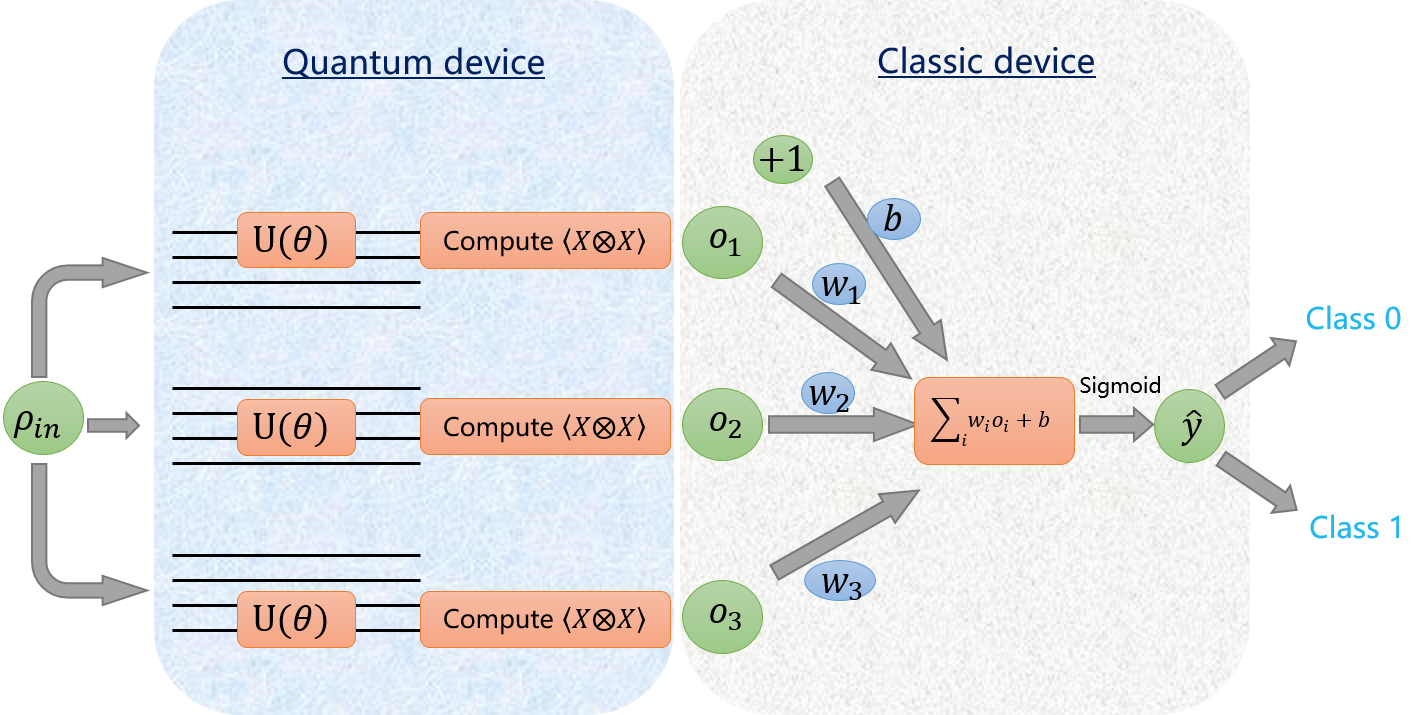}
	\caption{Sketch of variational shadow quantum learning (VSQL) for binary classification with $n=4$ and $n_{qsc} =2$. In the quantum device, the shadow circuit is implemented on the subspace of input state $\rho_{in}$. Sliding through the whole system to collect the Pauli-$(X\otimes X)$ expectations, i.e., shadow features. In the classic device, the resulting shadow features $o_i$'s are fed into a fully-connected neural network. Here, the output $\hat{y}$ is a value between 0 and 1 for binary case. We should denote that all the shadow circuits $U(\theta)$'s sliding through the $n$-qubit Hilbert space are identical.}
	\label{fig:scheme_vsqc}
\end{figure}
 
During the inference process, the unseen test data set $\mathcal{D}^{(test)}:=\{(\rho_{in}^{(m)}, y^{(m)}\in\{0,1\})\}_{m=1}^{N_{test}}$ is provided to the classifier $\mathcal{F}$. We feed each sample in the test set to the trained hybrid framework (combination of shadow circuit and FCNN) to predict its label. Then, the test accuracy could be calculated by comparing the predicted labels and the true labels. Due to the space limitation, we refer the details to Appendix (cf. Algorithm  \ref{alg:VSQL_binary_inference}).
Furthermore, VSQL can be naturally generalized to multi-label classification by replacing the sigmoid activation function with a softmax function. We defer the details for multi-label classification to Appendix.

\subsection{Loss function}
Given the data set  $\mathcal{D}:=\{(\rho_{in}^{(m)},
y^{(m)})\}_{m=1}^N$ and $n_{qsc}$-local shadow circuits, %where $y^{(m)}\in \{0, 1\}$ is the label of  $\rho_{in}^{(m)} \in \mathbb C^{2^{n}\times 2^n}$, 
the loss function of VSQL for binary classification is designed to be the mean square error\footnote{The cross-entropy loss is considered in the multi-label case; see Appendix.} \cite{Ziegel1999}:
\begin{small}
\begin{align}\label{eq:def_loss}
\mathcal{L}(\bm{\theta},\bm{w},b;\mathcal{D}) \!:=\! \frac{1}{2N}\!\sum_{m=1}^N \Lr{\hat{y}^{(m)}\lr{\rho_{in}^{(m)};\bm{\theta},\bm{w},b}\! -\! y^{(m)}}^2.
\end{align}
\end{small}
Here, the predicted label $\hat{y}^{(m)}$ is defined as follows:
\begin{small}
\begin{align}
    \hat{y}^{(m)}\lr{\rho_{in}^{(m)};\bm{\theta},\bm{w},b}\! :=\! \sigma \lr{\sum_{i} \! w_i o^{(m)}_i\lr{\rho_{in}^{(m)};\bm{\theta}}\!+\!b},
\end{align}
\end{small}
where $\sigma(z)$ denotes the sigmoid activation function and the shadow features $o_i$ are calculated through
\begin{small}
\begin{align}\label{eq:def_observ}
    o^{(m)}_i\lr{\rho_{in}^{(m)};\bm{\theta}} \!=\! \Tr \!\lr{ \rho_{in}^{(m)}  ( \mathbb I\!\otimes\! \cdots\! \otimes U^\dag(\bm{\theta}) {O} U(\bm{\theta})\otimes\! \cdots\! \otimes\! \mathbb I)  },
\end{align}
\end{small}
Note that the shadow circuit $U(\bm{\theta})$ and the physical observable ${O}=X\otimes\cdots\otimes X $ are applied on the same local qubits. Additionally, $U(\bm{\theta})$ is usually decomposed as a chain of unitary operators:
\begin{align}\label{eq:def_U}
    U(\bm{\theta})=\prod_{l=L}^1 U_l(\theta_l)V_l.
\end{align}
where $ U_l(\theta_l)=  \exp(-i\theta_lP_l/2)$ and $V_l$ denotes a fixed operator such as Identity, CNOT and so on.

\begin{algorithm}[t]
\caption{Variational shadow quantum learning (VSQL) for binary classification: the training process}
 \label{alg:VSQL_binary_train}
\begin{algorithmic}[1]
\REQUIRE The training data set $\mathcal{D}^{(train)}:=\{(\rho_{in}^{(m)}, y^{(m)}\in\{0,1\})\}_{m=1}^{N_{train}}$, $EPOCH$, optimization procedure
\ENSURE The final parameters $\bm{\theta}^*$, $\bm{w}^*$ and $b^*$, and the list of losses 
\STATE Initialize the parameters $\bm{\theta}$ of the 2-local (for example) shadow circuit $U(\bm{\theta})$ from uniform distribution Uni$[0, 2\pi]$ and $\bm{w},b$ from Gaussian distribution $\mathcal{N}(\bm{0},\mathbb I)$
\FOR{$ep=1,\ldots, EPOCH$}
    \FOR{$m = 1,\ldots, N_{train}$} 
    \STATE Apply multi-times the shadow circuit $U(\bm{\theta})$ to the input density matrix $\rho_{in}^{(m)}$
    \STATE Measure the subsystem and estimate a series of expectations $\langle X\otimes X \rangle$, recorded as $o_i$'s
    \STATE Feed the shadow features $o_i$'s into the classical neural network and obtain the output $\hat{y}^{(m)}$
    \STATE Compute the accumulated loss $\lr{\hat{y}^{(m)} -y^{(m)}}^2$ and update accordingly the parameters $\bm{\theta}$, $\bm{w}$ and $b$ via gradient-based optimization procedure
    \ENDFOR
    \IF{the stopping criterion is satisfied}
    \STATE Break
    \ENDIF
\ENDFOR
\end{algorithmic}
\end{algorithm}

\subsection{Analytical gradients}
%According to the loss function
With the above preparation, we can easily derive the analytical gradients, with which VSQL could naturally update its parameters $\bm{\theta}$ and $\{\bm{w}, b\}$ via gradient-based optimization method, e.g., SGD \cite{Bottou2004}. For each input $\rho_{in}^{(m)}$,
%For each data sample $(\rho_{in}^{(m)}, y^{(m)})$ in the data set $\mathcal{D}$, the partial derivatives with respect to the parameters $w_i,b$ and $\theta_l$ are calculated as follows:
\begin{small}
\begin{align}
    \frac{\partial \mathcal{L}}{\partial w_i} &= \lr{\hat{y}^{(m)} - y^{(m)}} \cdot \hat{y}^{(m)} \lr{1- \hat{y}^{(m)}} \cdot o^{(m)}_i,    \label{eq:derivative_w_i}        \\  
    \frac{\partial \mathcal{L}}{\partial b} &= \lr{\hat{y}^{(m)} - y^{(m)}} \cdot \hat{y}^{(m)} \lr{1- \hat{y}^{(m)}},   \label{eq:derivative_b}  \\
% \end{align}
% \begin{align}
    \frac{\partial \mathcal{L}}{\partial \theta_l} & = \frac{\partial \mathcal{L}}{\partial \hat{y}^{(m)}} \cdot \sum_i \frac{\partial \hat{y}^{(m)}}{\partial o^{(m)}_i} \cdot \frac{\partial o^{(m)}_i}{\partial \theta_l} \nonumber \\   \label{eq:derivative_theta} 
   & =  \lr{\hat{y}^{(m)} - y^{(m)}} \cdot \sum_i \hat{y}^{(m)} \lr{1- \hat{y}^{(m)}}  w_i \cdot \frac{\partial o^{(m)}_i}{\partial \theta_l},
\end{align}
\end{small}
%where $\hat{y}^{(m)}$ and $o^{(m)}_i$ are the corresponding abbreviations.
For the partial derivatives w.r.t $w_i$ and $b$ written in Eqs. \eqref{eq:derivative_w_i} and \eqref{eq:derivative_b}, they would be directly computed in the classical device and used to update $w_i,b$ through the back propagation algorithm \cite{Goodfellow-et-al-2016}. %of the classical neural networks.
For the partial derivative w.r.t $\theta_l$ in Eq. \eqref{eq:derivative_theta}, it can be regarded as a weighted sum of several partial gradients $ {\partial o^{(m)}_i}/{\partial \theta_l}$, 
%and each of which is calculated as, according to Eqs. \eqref{eq:def_observ} and \eqref{eq:def_U},
\begin{small}
\begin{align}\label{eq:ansatz_derivative_theta}
    \frac{\partial o^{(m)}_i\lr{\bm{\theta};\rho_{in}^{(m)}}}{\partial \theta_l} = -\frac{i}{2}\Tr \left(U_{> l}^\dag O U_{>l} \left[P_l, U_{\le l} \rho_i U_{\le l}^\dag\right] \right).
\end{align}
\end{small}
where $\rho_i=\Tr_{-i}(\rho_{in}^{(m)})$ denotes the partial trace of $\rho_{in}^{(m)}$ corresponding to the index $i$, $U_{\le l}=\prod_{j=l}^1 U_j(\theta_j)V_j$ and $U_{> l}=\prod_{j=L}^{l+1} U_j(\theta_j)V_j$ and $[\rho,\sigma]=\rho\sigma-\sigma\rho$ is a commutator. This gradient can be calculated exactly on the quantum device with the $\pi/2$ parameter shift rule proposed by \citeauthor{Mitarai2018}. Compared with the finite difference scheme, this method leads to a faster convergence \cite{Harrow2019} and is more suitable to the existing quantum devices.

\section{Theoretical Performance Analysis}
\label{sec:performance_VSQL}

\subsection{Number of parameters of VSQL}

In the hybrid quantum-classical framework, the number of parameters in the quantum circuit is an important metric to measure its complexity and efficiency. The main reason is that updating each parameter is costly in terms of quantum resources as it requires re-running the entire circuit multiple times. Therefore, algorithms with a smaller number of parameters are preferable in the NISQ era. Here, we will exhibit this advantage for VSQL.

There are two kinds of parameters in VSQL, i.e., the parameters $\bm{\theta}$ in the shadow circuits and the parameters $\bm{w},b$ in the classical NN. Assume the action mode of the shadow circuits is ``shadow sliding'' (illustrated in Fig. \ref{fig:scheme_vsqc}), which is also employed throughout this paper.
The number of parameters of VSQL for binary classification is summarized as follows.
\begin{prop}\label{prop:num_params}
For an $n$-qubit quantum system, if we use $n_s$ shadow circuits, then the number of parameters of VSQL for binary classification is 
\begin{align}
     \text{\# Params} &=  \text{\# Params} \big|_{\text{ in shadow circuits}} + \text{\# Params} \big|_{\text{ in NN}} \nonumber \\
     &= n_s n_{qsc} D +\Lr{ n_s \lr{n-n_{qsc}+1}+1},
\end{align}
where we denote by $n_{qsc}$ the \textbf{n}umber of \textbf{q}ubits of the \textbf{s}hadow \textbf{c}ircuits and assume each shadow circuit consists of $D$ layers with $n_{qsc}$ parameters in each layer. 
\end{prop}

Thus, VSQL has a parameter quantity that is linearly related to $n$ and $D$ separately, rather than $nD$ that commonly appears in most of the ansatzes employed in the existing literature \cite{Schuld2018,Farhi2018,Mitarai2018}.
For a 50-qubit quantum system, if we use just one 2-local shadow circuit with 20 layers, i.e., $n_s=1, n_{qsc}=2, D=20$, then the number of parameters of VSQL is $40+(50-2+1)+1=90$, which is much smaller than $nD=1000$.

\subsection{Theoretical classification ability}
\label{sec:theoretical_ability}

In this subsection, we explore the theoretical classification ability of VSQL and give the corresponding necessary and sufficient conditions.

\begin{thm}\label{thm:classi_ability}
Given two types of input density matrices $\rho_{in}^{(0)}$ and $\rho_{in}^{(1)}$ with labels 0 and 1, respectively, if there exists a group of $\bm{\theta} $ that makes at least one pair of shadow features $o_i^{(0)}$ and $o_i^{(1)}$ different, i.e., $| o_i^{(0)}-o_i^{(1)}|  >0 $, then VSQL can distinguish them, vice versa.
\end{thm}

\begin{coro}\label{coro:classi_ability_partial_trace}
Given two types of $n$-qubit input density matrices. If each pair of their corresponding $m$-local partial traces are identical ($m<n$), then VSQL is theoretically incapable of distinguishing them via $m$-local shadow circuits, and vice versa.
\end{coro}
The proof of Corollary \ref{coro:classi_ability_partial_trace} could be immediately derived from Theorem \ref{thm:classi_ability}, because getting identical partial traces is equivalent to having same shadow features (cf. Eq. \eqref{eq:def_observ}).

% \subsubsection{Classification ability with different local shadow circuits}

After exploring the necessary and sufficient condition for the theoretical classification ability of VSQL, we now discuss this ability under different local shadow circuits. Intuitively, larger shadow circuits will give VSQL stronger classification ability. The following Theorem will give a detailed statement.

\begin{thm}\label{thm:n_1_n_2_lcoal_shadow}
Given two types of $n$-qubit input density matrices $\rho_{in}^{(0)}$ and $\rho_{in}^{(1)}$. If VSQL can not theoretically distinguish them via $m$-local shadow circuits, then neither can via $m^\prime$-local shadow circuits, where $m^\prime<m<n$. And not vice versa.
\end{thm}

The proof is shown in Appendix. From Theorem \ref{thm:n_1_n_2_lcoal_shadow}, we confirm the intuition that the larger the number of qubits $n_{qsc}$ of the shadow circuits is, the stronger the theoretical expressive ability of VSQL is. However, if this number is too large, it will lead to other problems, such as the Barren Plateau issue described in the next subsection.

\subsection{ Escape of Barren Plateau}
In the last subsection, we have shown that VSQL has a strong theoretical classification ability for a wide range of quantum states, especially by using large local shadow circuits. However, the sizeable operating scope of the shadow circuits will increase network parameters and the cost of compiling given limited hardware connections, and leads to the Barren Plateau (BP) issue. The BP issue \cite{McClean2018,Cerezo2020} refers to the vanishing gradient problem during the training process of QNN. That is, for a wide range of variational quantum circuits, the partial gradients of the objective function have a zero mean and an exponentially vanishing variance, which makes it difficult for the optimizer to find the correct direction to decrease the objective function. Therefore, it is important to discuss whether the BP problem exists when proposing a new variational quantum algorithm.

Next, we will evaluate the mean and the variance of the analytical gradients in VSQL. There is no BP issue for the partial gradients (see Eqs. \eqref{eq:derivative_w_i} and \eqref{eq:derivative_b}) with respect to the parameters $w_i$ and $b$ of the classical NN. And for the partial gradient (see Eq. \eqref{eq:derivative_theta}) with respect to $\theta_l$ of QNN, the BP issue is mainly reflected on the last term, i.e., the partial derivatives (see Eq. \eqref{eq:ansatz_derivative_theta}) of the shadow features $o_i$ with respect to $\theta_l$. Hence, it is sufficient to evaluate the mean and the variance of the partial gradient in Eq. \eqref{eq:ansatz_derivative_theta} to explore the BP problem in VSQL.
The results are summarized in the following proposition.

\begin{defn}
A unitary  $t$-design \cite{Dankert2009}  is defined as a finite set of unitaries $\{U_k\}_{k=1}^K$ on a d-dimensional Hilbert space
such that 
% \begin{small}
\begin{align}
\frac{1}{K}\cdot \sum_k P_{(t,t)}(U_k) = \int_{\mathcal{U}(d)} d\mu_{Haar}(U) P_{(t,t)}(U),
\end{align}
% \end{small}
where $P_{(t,t)}(U)$ denotes a polynomial of degree at most $t$ on the elements of $U$ and at most $t$ on the elements of $U^\dag$.
\end{defn}

\begin{prop}\label{prop:barren_VSQL}
If $U_{> l}$ or  $U_{\le l}$ forms at least an $n_{qsc}$-local unitary 2-design,
the mean and the variance of the analytical gradients with respect to $\theta_l$ in VSQL (see Eq. \eqref{eq:ansatz_derivative_theta})  are evaluated as
\begin{align}
    \mathbb E\Lr{\frac{\partial o_i}{\partial_{\theta_l}}} &=0;\quad
    \text{Var}\Lr{\frac{\partial o_i}{\partial_{\theta_l}}} = -\frac{1}{4}\cdot \frac{C \lr{\rho_i} }{2^{2n_{qsc}}-1},
\end{align}
where $C\lr{ \rho_i}\in \lr{- 4 \times 2^{n_{qsc}}, 0 }$ denotes a constant  and 
$n_{qsc}$ is the number of qubits of the shadow circuits.
\end{prop}

The proof is shown in Appendix. From Proposition \ref{prop:barren_VSQL}, we notice the variance of the gradients decays exponentially with $n_{qsc}$, rather than the qubit number $n$. Hence, no matter how big the problem size $n$ is,
as long as we choose a small  $n_{qsc}$ (e.g., $n_{qsc}\le 4$) and assume $C\lr{ \rho_i} \approx - 2\times 2^{n_{qsc}}$, we can evaluate the analytical gradients efficiently via more than 1,000 repetitions derived from the Chernoff bound.
In one word, VSQL could escape the barren plateaus by choosing an appropriate operating scope $n_{qsc}$. 
Moreover, \citeauthor{Wang2020} indicates that noise could also induce the BP issue. Following this line of reasoning, the small shadow circuits in VSQL, where less noise is introduced, will also be beneficial for escaping the barren plateaus from a different perspective.

Here, we provide an illustrated example. Assume the $n$-qubit quantum state $\rho_{in}=\op{\psi_{in}}{\psi_{in}}$ we want to classify is labeled with 0, where
\begin{align}
\ket{\psi_{in}} := &\otimes_{j=0}^{n-1} R_y(2\pi j/n)\ket{0}.
\end{align}
The chosen shadow circuit consists of a layer of single-qubit $R_y$ rotations and a layer of CNOT gates which only connects the adjacent qubits, followed by another layer of $R_y$ rotations. Then, we compute the loss landscape in Eq. \eqref{eq:def_loss} with regard to the first two circuit parameters by fixing all the other parameters with $\pi/4$ and setting the bias $b=0$ and $\bm{w} \sim \mathcal{N}(\bm{0},\mathbb I)$ sampled from a Gaussian distribution.
The result, as shown in Fig. \ref{fig:bp_landscape1}, is in line with the above analysis, i.e., there is no barren plateaus with $n_{qsc}=2$, but the loss landscape shrinks dramatically with an increasing $n_{qsc}$.

\begin{figure}[t]
\centering
%%% ===== subfig-1 ===== 
\subfigure[$(2,2)$]{\label{fig:bp_n2qsc2}
\includegraphics[width=0.1\textwidth,height=1.5cm]{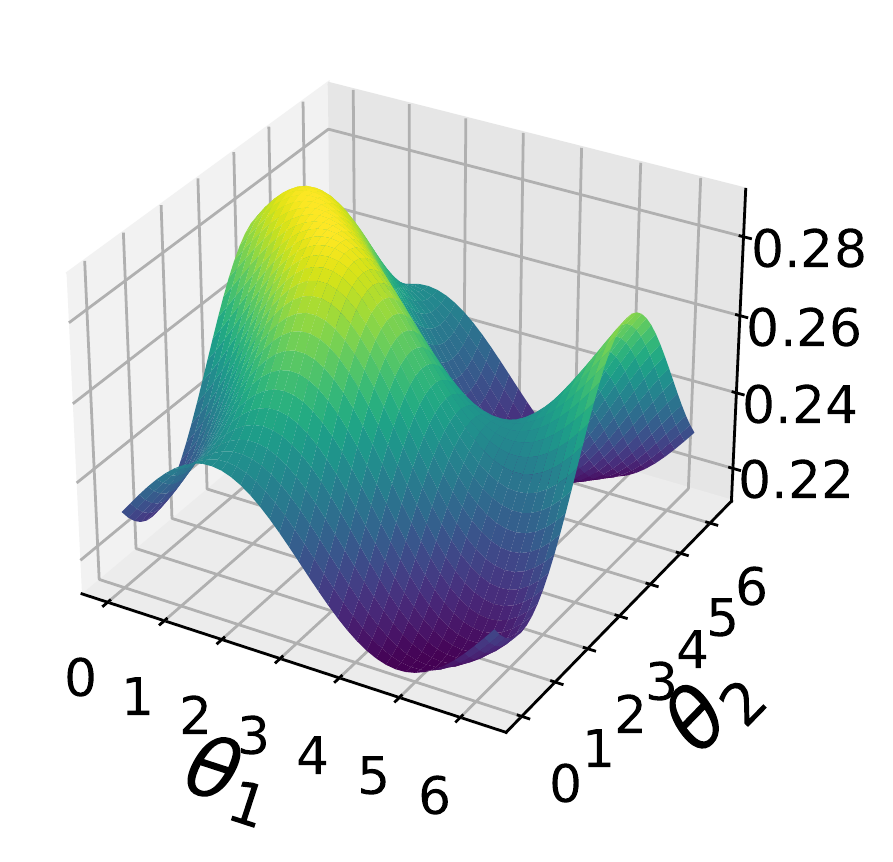}}
%%% ===== subfig-2 ===== 
\subfigure[$(10,2)$]{\label{fig:bp_n10qsc2}
\includegraphics[width=0.11\textwidth,height=1.5cm]{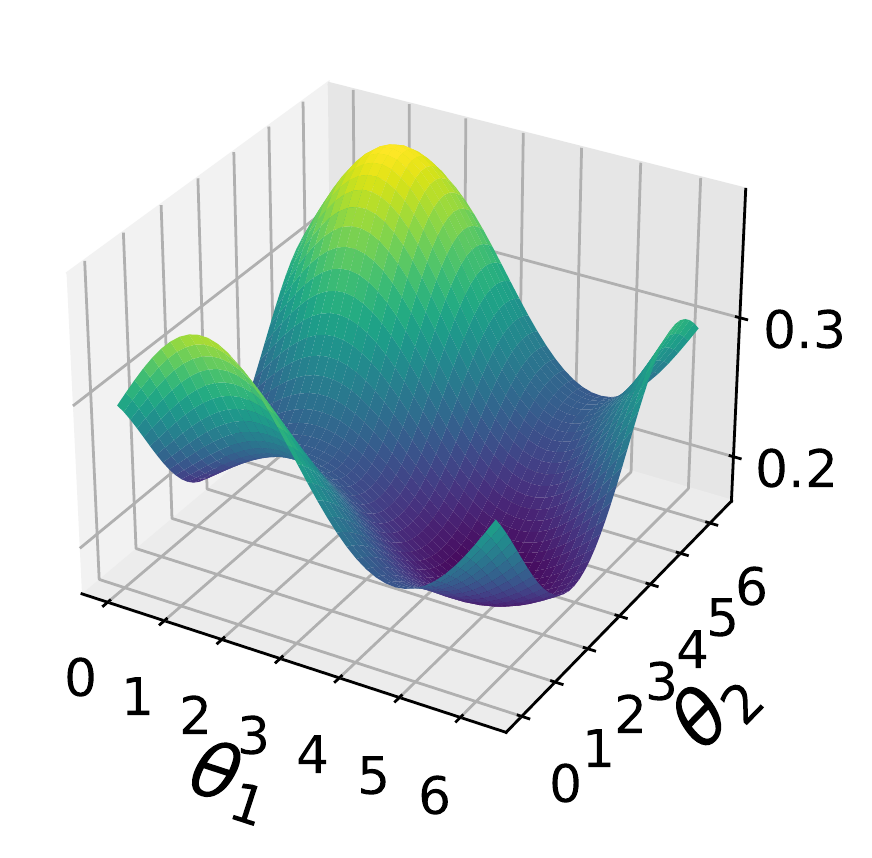}}
%%% ===== subfig-3 ===== 
\subfigure[$(20,2)$]{\label{fig:bp_n20qsc2}
\includegraphics[width=0.11\textwidth,height=1.5cm]{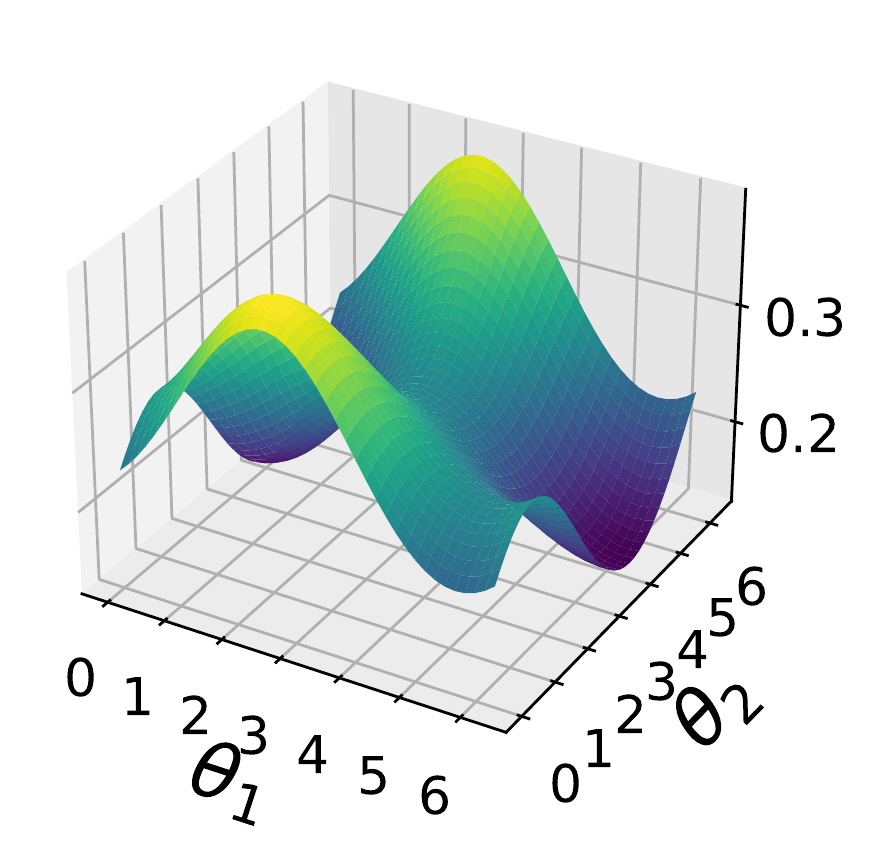}}
%%% ===== subfig-4 ===== 
\subfigure[$(100,2)$]{\label{fig:bp_n100nqsc2}
\includegraphics[width=0.11\textwidth,height=1.5cm]{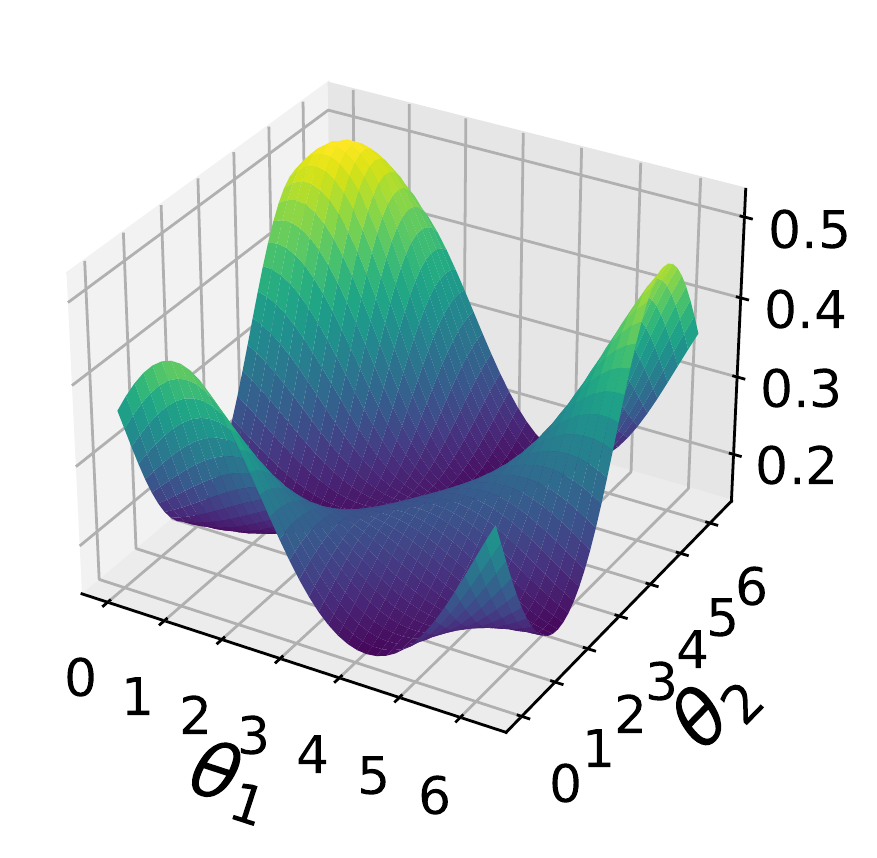}}
%%% ===== subfig-5 ===== 
\subfigure[$(10,2)$]{\label{fig:bp_n10nqsc2}
\includegraphics[width=0.1\textwidth,height=1.5cm]{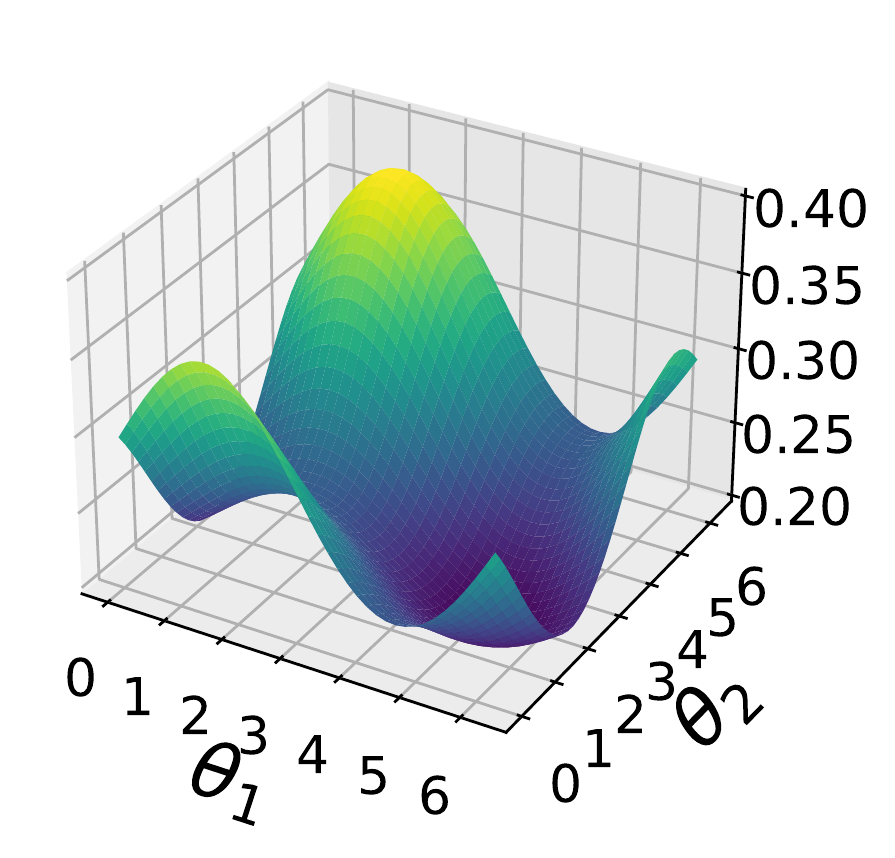}}
%%% ===== subfig-6 ===== 
\subfigure[$(10,4)$]{\label{fig:bp_n10qsc4}
\includegraphics[width=0.11\textwidth,height=1.5cm]{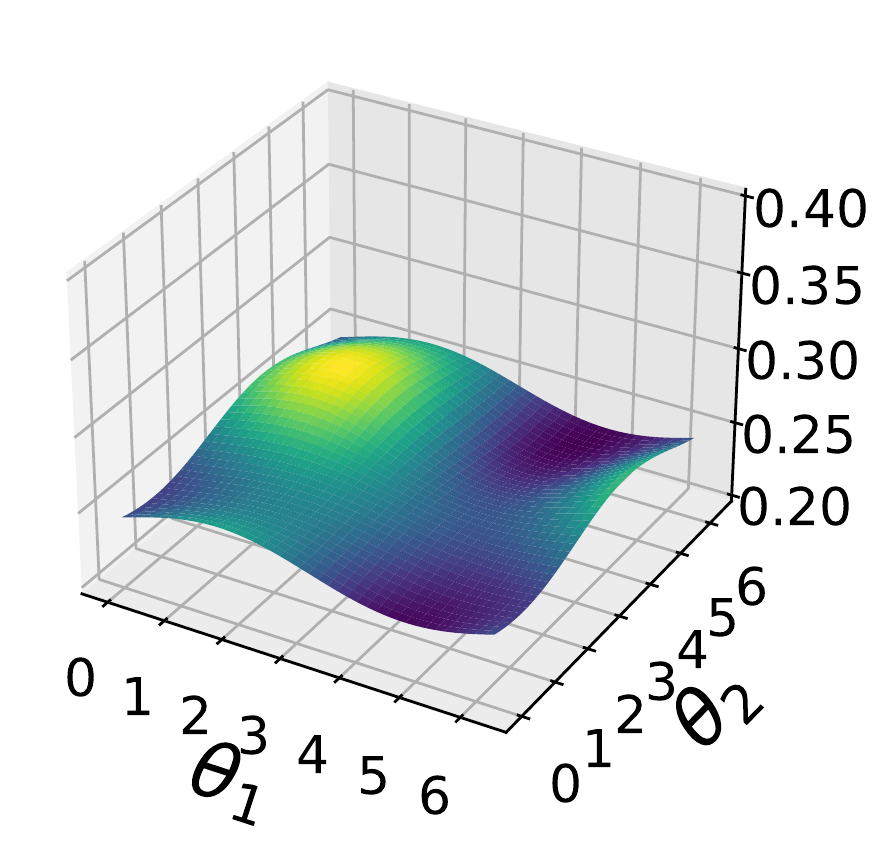}}
%%% ===== subfig-7 ===== 
\subfigure[$(10,6)$]{\label{fig:bp_n10qsc6}
\includegraphics[width=0.11\textwidth,height=1.5cm]{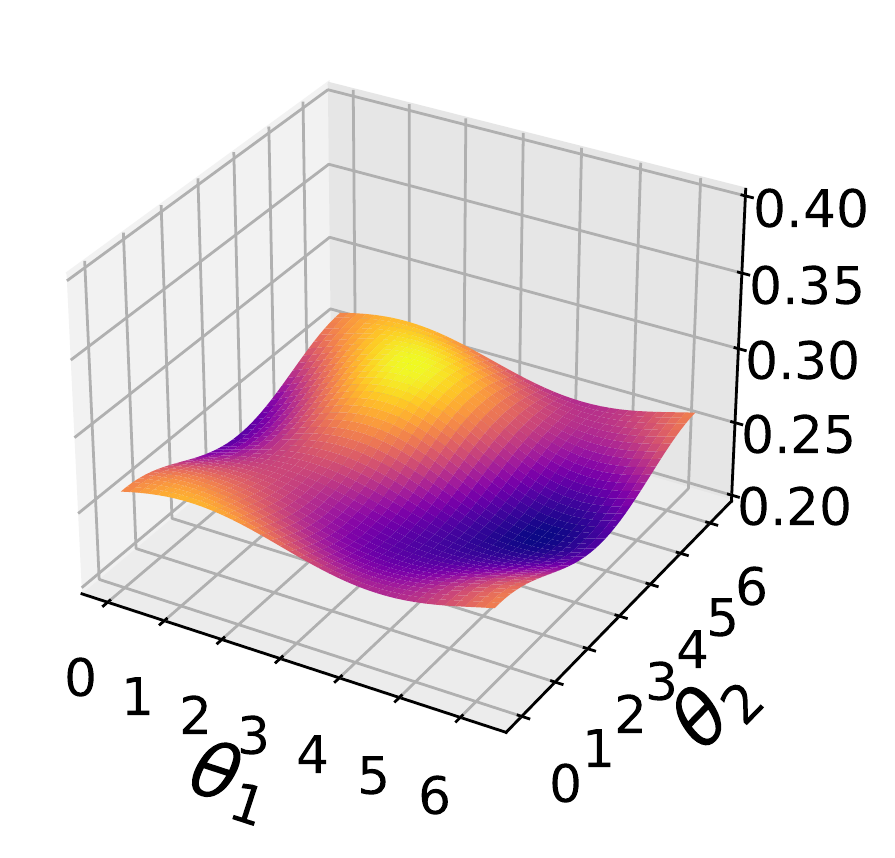}}
%%% ===== subfig-8 ===== 
\subfigure[$(10,10)$]{\label{fig:bp_n10qsc10}
\includegraphics[width=0.11\textwidth,height=1.5cm]{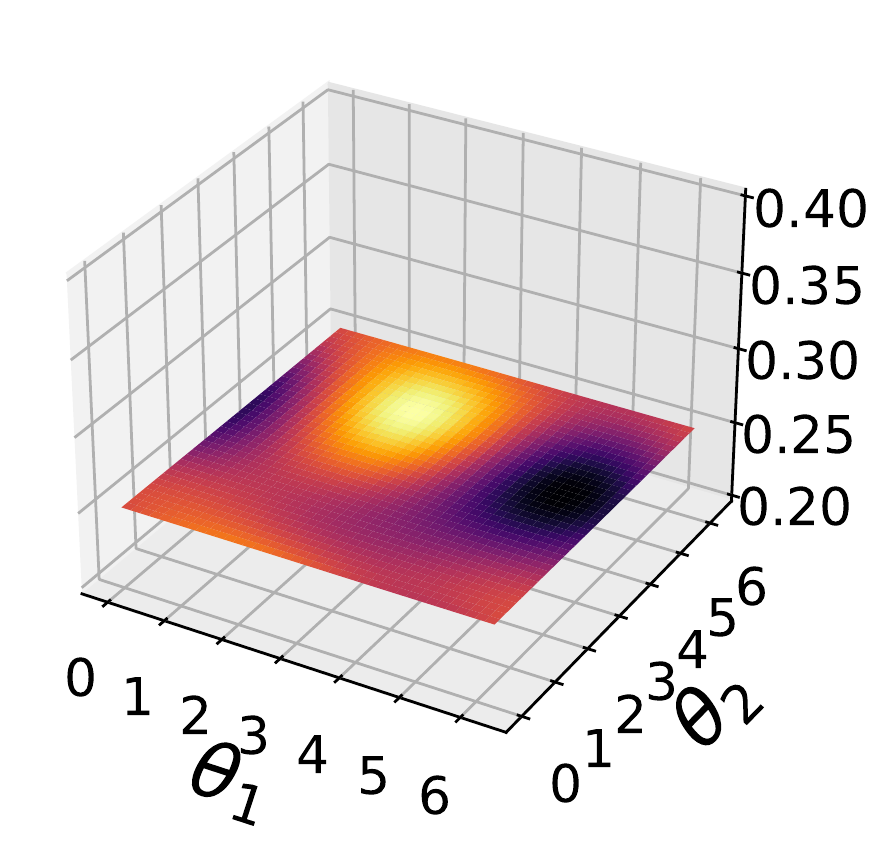}}
\caption{The slice of loss landscape with respect to the first two circuit parameters by changing the system size $n$ and operating scope $n_{qsc}$. Here, the binary list represents $(n,n_{qsc})$.}
%\caption{The loss landscape with respect to the first two parameters in Fig. \ref{fig:ansatz_mnist_binary} under different $n$ and $n_{qsc}$. (a)-(d) have the same $n_{qsc}=2$ but different $n$; (e)-(h) have the same $n=10$ but different $n_{qsc}$.}
\label{fig:bp_landscape1}
\end{figure}

\section{Numerical Experiments}
\label{sec:appli_VSQL}

We supplement our theoretical results with numerical experiments by classical simulation of VSQL. Specifically, our numerical experiments include distinguishing two (and three) families of $2$-qubit quantum states and classifying handwritten digit images taken from the MNIST data set.
We also conduct experiments on classifying noisy quantum states to exhibit the robustness of VSQL, which is deferred to Appendix due to limited space.
All the simulations and optimization loop are implemented via Paddle Quantum\footnote{https://github.com/paddlepaddle/Quantum} on the PaddlePaddle Deep Learning Platform~\cite{Ma2019p}.
% All the simulations and optimization loop are implemented via Paddle Quantum~\cite{Paddlequantum} on the PaddlePaddle Deep Learning Platform~\cite{Paddle,Ma2019p}.

\subsection{Classification of quantum states}
Quantum state discrimination (QSD) is a fundamental problem in quantum information and it underlies various applications~\cite{Nielsen2002,Barnett2009,Bae2015}. The goal of QSD is to determine which state a given physical system has been prepared. In general, finite copies of the given states are considered in the study of distinguishing quantum states~\cite{Lloyd2020,Gambs2008}.

\subsubsection{ Classification of binary quantum states}
We choose two canonical families of non-orthogonal 2-qubit quantum states as a proof of principle. These states are well-studied in Refs. \cite{Mohseni2004,Chen2018,Patterson2019}, and are parametrized by real numbers $u$ and $v$. Here, we use the Dirac (bra-ket) notation to represent the quantum states as
%The quantum states to be distinguished, we choose here, are two canonical families of non-orthogonal 2-qubit quantum states, which are analyzed in detail as well in \cite{Mohseni2004,Chen2018,Patterson2019}. These two families of states are identified by $u$ and $v$:
% \begin{small}
\begin{align}
    \ket{\psi_u} &= [\sqrt{1-u^2}, 0, u, 0]^\top,   \label{eq:psi_u} \\
    \ket{\psi_{v \pm}} &= [0, \pm \sqrt{1-v^2}, v, 0]^\top,   \label{eq:psi_v}
\end{align}
% \end{small}
where %the two real hyper-parameters 
$u,v \in [0, 1]$.
Then, we can write this two sets of quantum states as a mixed quantum state $\rho$:
\begin{small}
\begin{align}
\label{eq:quantum_data_set}
    \rho(u,v)\! :=\!  q_1 \underbrace{\op{\psi_u}{\psi_u}}_{\rho_1(u)} \!+
     \frac{q_2}{2} (\underbrace{\op{\psi_{v+}}{\psi_{v+}}
     \! +\! \op{\psi_{v-}}{\psi_{v-}}}_{\rho_2(v)}),
\end{align}
\end{small}
with probability distribution $\{q_1\!=\!\frac{1}{3}, q_2\!=\!\frac{2}{3}\}$. These choices are consistent with the existing literature \cite{Mohseni2004,Chen2018}. 
% and $\rho_2=\op{\psi_v}{\psi_v} := \frac{1}{2}\lr{ \op{\psi_{v+}}{\psi_{v+}} +\op{\psi_{v-}}{\psi_{v-}}}$, where the sign $v+$ and $v-$ are corresponding to the sign $\pm$ of the second component of $\ket{\psi_v}$  in Eq. \eqref{eq:psi_v}.

\subsubsection{Theoretical distinguishability}
We first analyze the ability of our method for classifying these two families of quantum states. The result is summarized in Theorem \ref{thm:distinguish_psi_uv}.
\begin{thm}\label{thm:distinguish_psi_uv}
Given two families of non-orthogonal 2-qubit quantum states, shown in Eq. \eqref{eq:quantum_data_set}, and each has multiple copies. VSQL could exactly distinguish them by using only one shadow circuit, which consists of only one $R_y$ rotation gate applied on a 1-local qubit.
\end{thm}
We defer the proof to Appendix. This Theorem shows that VSQL could theoretically distinguish these two different families of quantum states. We further evaluate the performance of VSQL via numerical experiments.

\subsubsection{Experimental setting and results}
300 density matrices with 100 $\rho_1(u)$ (labeled as `0') and 200 $\rho_2(v)$ (labeled as `1') are sampled according to Eq. \eqref{eq:quantum_data_set}, where the parameters $u$ and $v$ are uniformly taken from $[0, 1]$. Then, we randomly select 80\% of them as the training set and the rest 20\% as the validation set. Consistent with the Theorem above, one 1-local shadow circuit, which consists of $R_y$ gate only, is used to extract local features. The parameters of the shadow circuit $\bm{\theta}$ and the FCNN $\{\bm{w}, b\}$ are initialized from the uniform distribution $\text{Uni}[0,2\pi]$ and the Gaussian distribution $\mathcal{N}(\bm{0},\mathbb I)$, respectively. During the optimization loop,  we choose the Adam \cite{Kingma2014} optimizer with a learning rate $\text{LR}=0.03$. 
%During the training process, 
Learning curves for the training loss and the validation accuracy are illustrated in Fig. \ref{fig:qsd_binary10}, where the distinguishability shown coincides with Theorem \ref{thm:distinguish_psi_uv}. We conclude VSQL could perfectly recognize the two families of quantum states defined in  Eq. \eqref{eq:quantum_data_set} after about 700 iterations. We find that the classification task becomes very difficult when $u,v$ are both close to 1. This makes sense because $\rho_1(u=1) \rightarrow \rho_2(v=1)$ on the extreme case. This experimental result highlights the strength of our method. As a comparison, we adjust the sample range such that e.g., $u,v\in [0.1,0.9]$, and the results are shown in Fig. \ref{fig:qsd_binary09}. 
As expected, this modification leads to a faster convergence.

 \begin{figure}[t]
	\centering
	\subfigure[$ u,v\in \text{Uni}\Lr{0,1}$]{\label{fig:qsd_binary10}
		\includegraphics[width=0.22\textwidth]{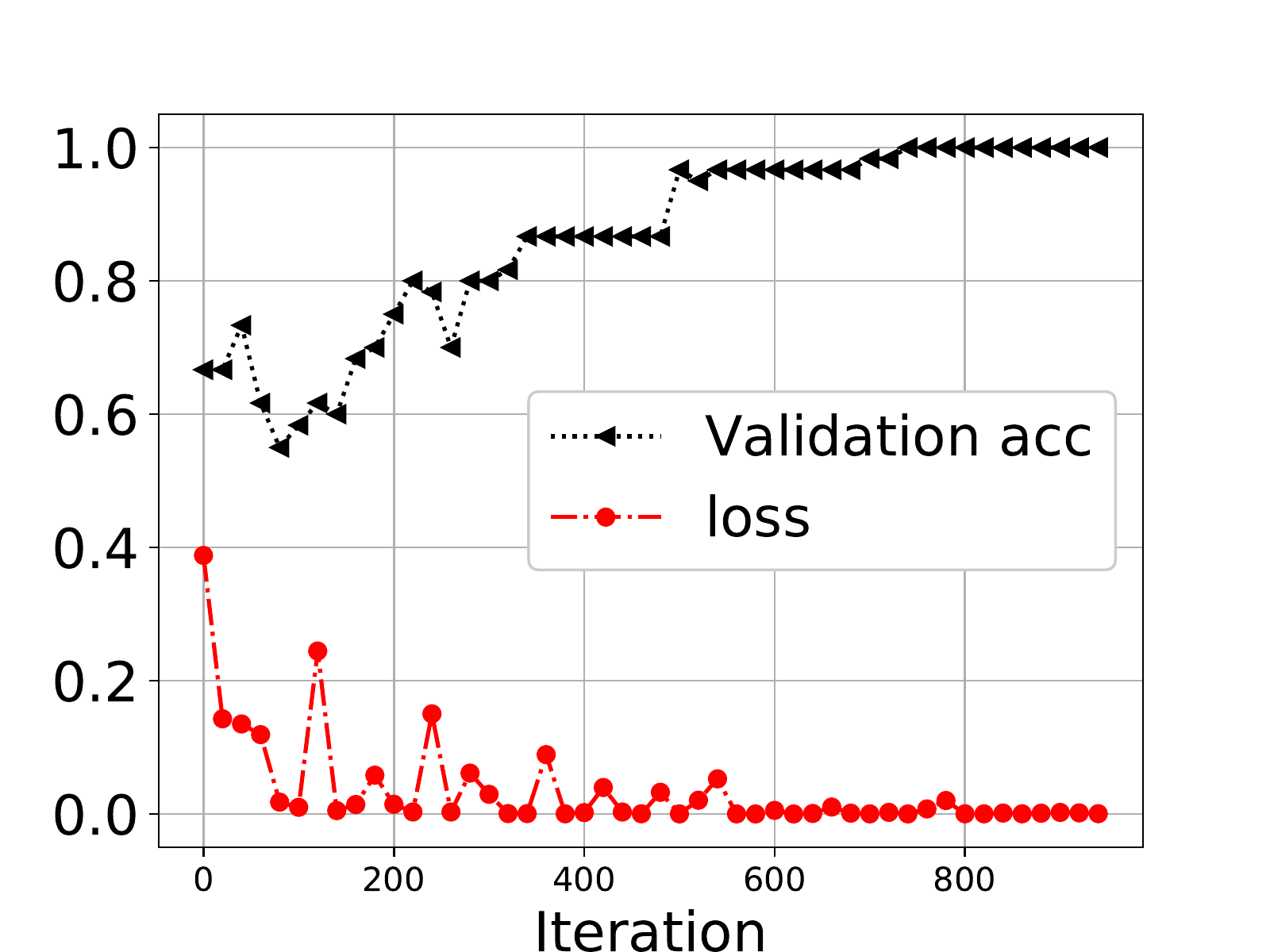}}
	\subfigure[$ u,v\in \text{Uni}\Lr{0.1,0.9}$]{\label{fig:qsd_binary09}
		\includegraphics[width=0.22\textwidth]{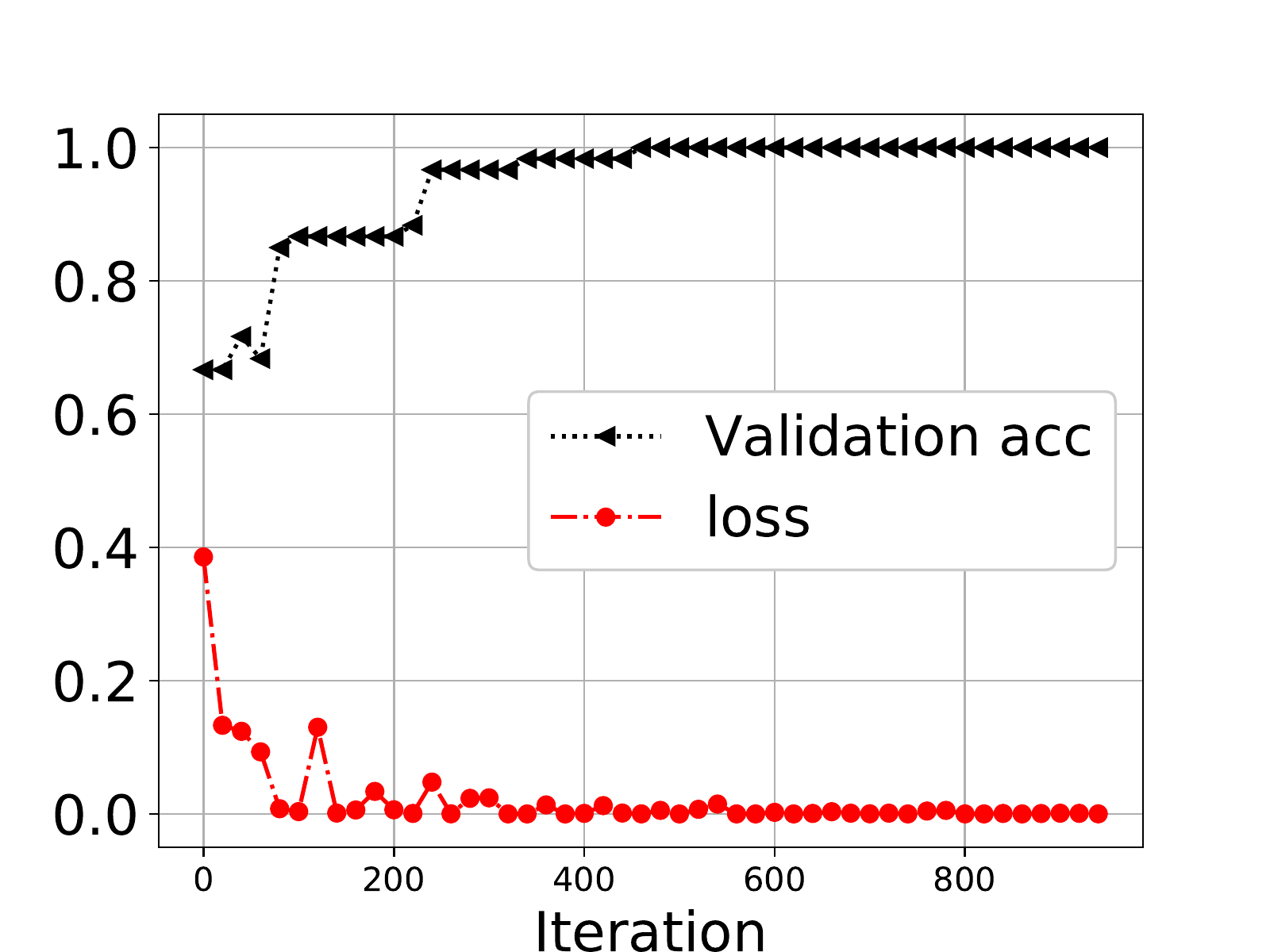}}
	\subfigure[$ u,v,t\in \text{Uni}\Lr{0,1}$]{\label{fig:qsd_multi10}
		\includegraphics[width=0.22\textwidth]{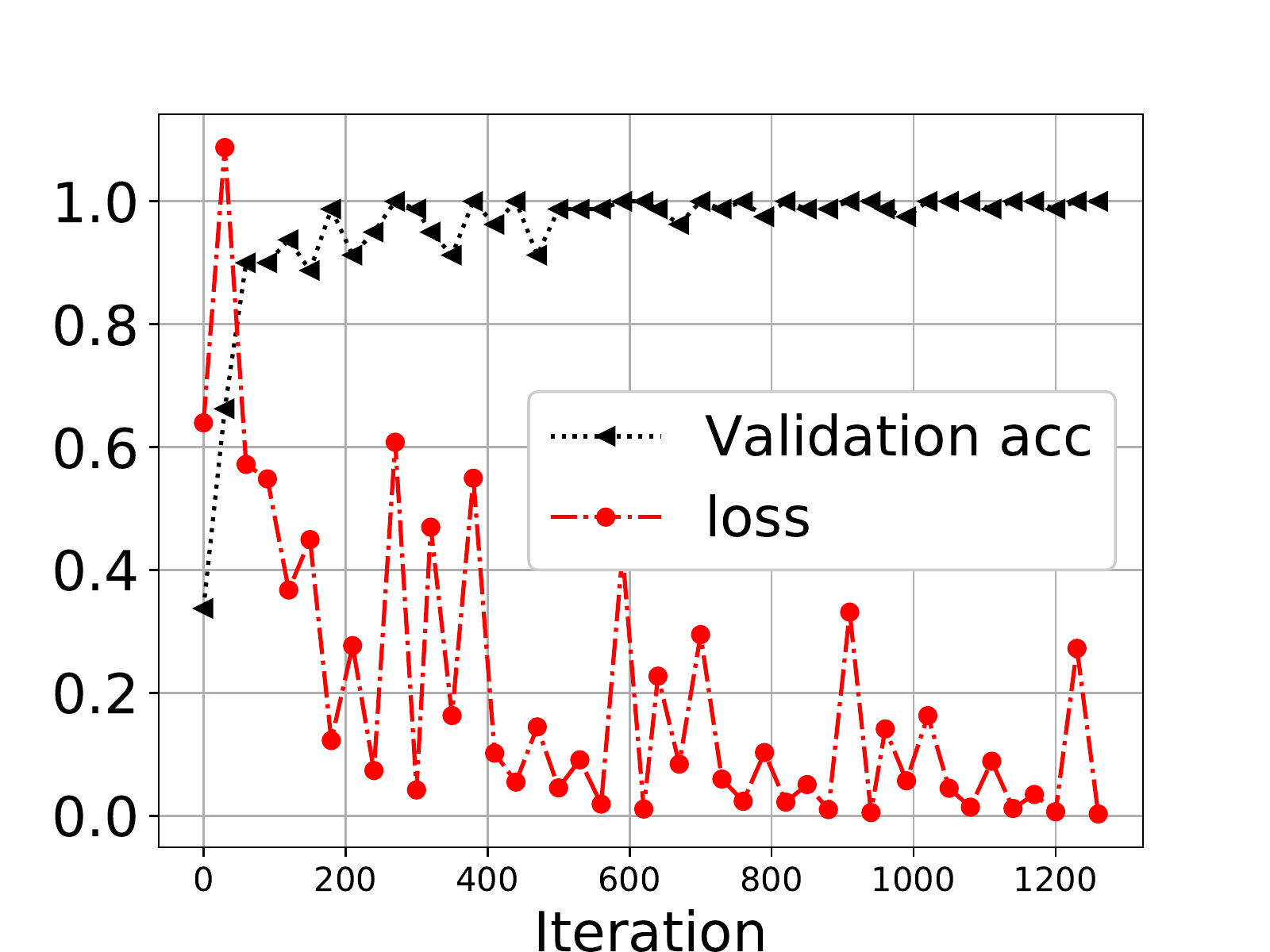}}
	\subfigure[$ u,v,t\in \text{Uni}\Lr{0.1,0.9}$]{\label{fig:qsd_multi09}
		\includegraphics[width=0.22\textwidth]{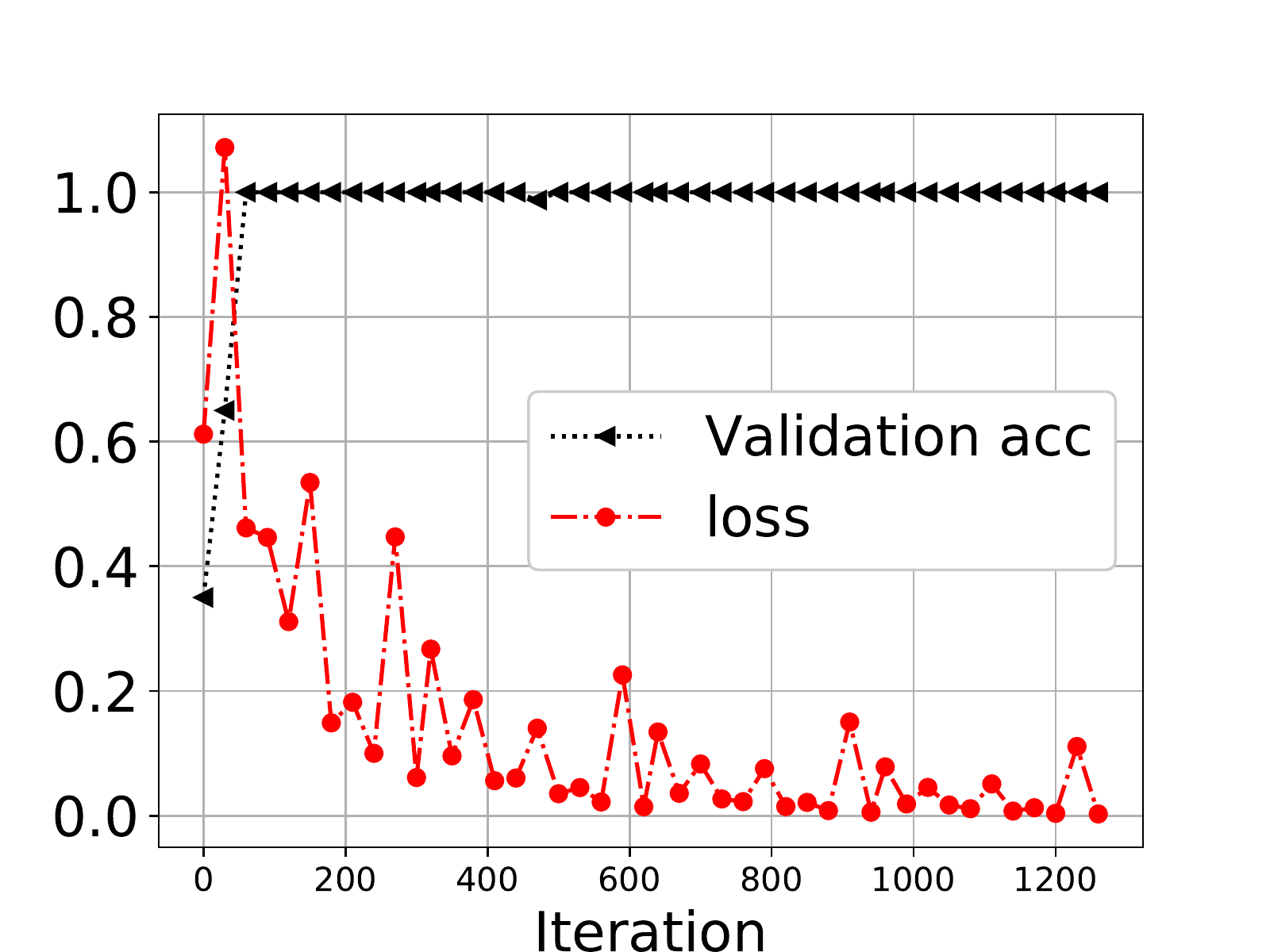}}
	\caption{Learning curves that record the training loss and the validation accuracy of VSQL with different experimental settings. (a) and (b) are binary classification with different parameter range $u,v \in [0,1]$ and $u,v \in [0.1,0.9]$. By adjusting the sample range, the training loss and the validation accuracy reach the optimal values faster. (c) and (d) describe the similar experimental setting but for three-class classification of quantum states.}
	\label{fig:qsd_class}
\end{figure}

\subsubsection{Classification of multi-class quantum states}
We also conduct the experiments for three-class case with 400 density matrices in total, where 
the data set is again taken from Eq. \eqref{eq:quantum_data_set}, but adding the third family $\rho_3(t) = \op{\psi_t}{\psi_t}$ with a new probability distribution $\{q_i\} = \{\frac{1}{4}, \frac{1}{2}, \frac{1}{4} \}$, where 
\begin{align} \label{eq:psi_t}
    \ket{\psi_t} = [\sqrt{1-t^2}, t, 0, 0]^\top, \quad t\in [0,1].
\end{align}
%  $t\in [0,1]$.
The experimental settings are almost the same as the binary case and results are illustrated in Figs. \ref{fig:qsd_multi10}  and \ref{fig:qsd_multi09}.  
We conclude that VSQL also has a strong distinguishability for multi-class quantum states. A complete analysis about this block could be found in Appendix.

\subsection{MNIST classification}
Next, we apply VSQL to classify handwritten digits taken from a public benchmark dataset MNIST \cite{LeCun1998}, which consists of 60,000 train examples and 10,000 test examples.
The MNIST data set contains 10 different classes labeled from `0' to `9'. Each image contains 28 $\times$ 28 gray scale pixels valued in $0\sim 255$. In order to match the input of VSQL, these pictures are normalized and unfolded into 784-dimensional vectors. Then, we expand their dimension to 10-qubit pure quantum states $\{\ket{\psi_i}\}$ (1024-dimensional vectors) via zero-padding and represent them in the density matrix formulation $\{\rho^{(i)}_{in}\} = \{\ket{\psi_i}\bra{\psi_i}\}$. By doing so, the pre-processing is complete and we obtain the training set $\mathcal{D}^{(train)}:=\{(\rho_{in}^{(m)}, y^{(m)})\}_{m=1}^{N_{train}}\subset \mathbb C^{1024\times 1024}\times \mathbb R^{10}$.
We first select two classes (`0' and `1') to verify the binary classification ability of VSQL, which contains 12,665 training samples (5923 0-label and 6742 1-label) and 2115 test samples (980 0-label and 1135 1-label). Then, we use the whole data set to evaluate the 10-class classification performance. 

\subsubsection{Experimental setting}
For the binary case, the 2-local shadow circuit (ansatz) used to extract local features is shown in Fig. \ref{fig:ansatz_mnist_binary}. The number of repetitions of the dashed block structure is denoted as the circuit depth $D$ and this ansatz has $2(D+3)$ parameters in total. The parameters $\bm{\theta}$ and $\{\bm{w},b\}$ are initialized from a uniform distribution in $[0, 2\pi]$ and a Gaussian distribution $\mathcal{N}(\bm{0},\mathbb I)$, respectively. 
During the optimization, we choose the Adam optimizer with a batch size of 20 samples and a learning rate of $\text{LR} =0.02$.
% During the optimization loop, we choose the Adam (Adaptive moment estimation) \cite{Kingma2014} optimizer with a batch size of 20 samples. 
%The learning rate is set to be $\text{LR} =0.02$. 
Each experiment is repeated 10 times to collect the mean accuracy and the corresponding fluctuations.
%For the binary case, the ansatz is chosen as the 2-local shadow circuit shown in Fig. \ref{fig:ansatz_mnist_binary}, where $R_z-R_y-R_z$ represents rotation framework is first selected, followed by repeated two CNOTs and two single-qubit $R_y$ rotations. Assume the number of repeated times is $D$, which is also called depth, then there are totally $2(D+3)$ parameters in the ansatz. Here we shall know that in order for a better performance, we could take multiple ansatzes simultaneously with different parameters. The parameters in the ansatz and FCNN are initialized from a uniform distribution in $[0, 2\pi]$ and a Gaussian distribution $\mathcal{N}(\bm{0},\mathbb I)$, respectively, and updated via ADAM method \cite{Kingma2014} with batches of 20 samples. The learning rate is set as 0.02. Each experiment is running 10 times to obtain the mean and the corresponding fluctuation.
For the 10-class case, the classification task becomes much difficult and hence we choose 4-local shadow circuits to extract shadow features, which can be extended from the 2-local design in Fig. \ref{fig:ansatz_mnist_binary}. There will be $4(D+3)$ parameters in each shadow circuit. All the other settings are identical to the binary case, except for a new batch size of 200 samples.
%For the 10-class case, due to the difficulty of multiple classification task, we employ 4-local shadow circuits, which are direct extensions of 2-local shadow circuits, to extract shadow features. Hence there are $4(D+3)$ parameters in each shadow. All the other settings are the same as the binary case, except the batches changing to 200 samples.

  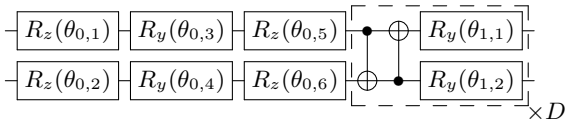
\begin{figure}[t]
  \small
\[\Qcircuit @C=0.5em @R=0.5em {
   & \gate{R_z(\theta_{0,1})} & \gate{R_y(\theta_{0,3})} & \gate{R_z(\theta_{0,5})} &\ctrl{+1}& \targ &  \gate{R_y(\theta_{1,1})}&\qw 
%   &\cdots& &\ctrl{+1}& \targ & \gate{R_y(\theta_{D,1})} &\qw 
   \\ 
   & \gate{R_z(\theta_{0,2})} & \gate{R_y(\theta_{0,4})} & \gate{R_z(\theta_{0,6})} & \targ & \ctrl{-1} &  \gate{R_y(\theta_{1,2})} & \qw 
%   &\cdots& &\targ & \ctrl{-1} & \gate{R_y(\theta_{D,2})}    &\qw 
   \gategroup{1}{5}{2}{7}{0.5em}{--} \\
  &&&&&&& \ \ \ \ \times D 
}\]
\caption{The 2-local shadow circuit design for MNIST classification (binary case). The first part uses $R_z-R_y-R_z$ combination to represent general rotations on each single-qubit subspace. The followed repeated block consists of CNOT gates  and two single-qubit $R_y$ rotations. The block circuit in the dashed box is repeated $D$ times to extend the expressive power of quantum circuits.}
\label{fig:ansatz_mnist_binary}
\end{figure}

\subsubsection{Results}
The results for the binary case are summarized in Table \ref{table:result_mnist_binary}. Our method VSQL easily achieves an average test accuracy above 99\% with only $n_s =1$ shadow circuit and depth $D = 1$, which has 8 rotation angles in the shadow circuit and 9 weights and 1 bias in FCNN. This result demonstrates the powerful capacity of VSQL to classify handwritten digits. By adding another shadow circuit to $n_s = 2$ with 35 parameters, one could obtain an average test accuracy above 99.5\%. As a comparison, we list the results of existing methods: Circuit-centric classifier \cite{Schuld2018} and QNN classifier \cite{Farhi2018}. Our method outperforms these variational quantum classifiers in terms of the number of parameters and test accuracy. Here, we should note that the details of their data preprocessing are slightly different from us, i.e., the Circuit-centric classifier uses the MNIST256 dataset with an 8-qubit quantum system, and the QNN classifier uses a $4\times 4$ downsampled version of the MNSIT dataset with a 17-qubit quantum system. 

%Here we merely give a rough reference and conclude that VSQL could obtain a higher accuracy while requiring fewer parameters.

% \linespread{1.5}
\begin{table}[t]
\small
%\caption{MNIST binary classification}
\centering
\begin{tabular}{c|c|c|c|c}
\toprule
 Methods & $n_s$ & $D$  & \# Ps & Test acc (\%) \\ \midrule
%  Circuit-centric classifier
 \cite{Schuld2018} &  /  & /  & 124   &   96.70  \\ \midrule
%  QNN classifier
 \cite{Farhi2018} &  /  &  / & 96   &    98.00  \\ \midrule
%  Hierarchical classifier \cite{Grant2018} &  /  &  / &  > 66   &    99.74 $\pm$ 0.02  \\ \midrule
 \multirow{2}{*}{VSQL (this paper)} & \multirow{1}{*}{1} & 1       &            18       & 99.43 $\pm$ 0.14 \\ \cline{2-5} 
        & \multirow{1}{*}{2} & 1   &    35             &  \textbf{99.52 $\pm$ 0.18} \\ \toprule
\end{tabular}
% \linespread{1.0}
\caption{Summary of the existing variational quantum classifiers on MNIST binary classification. VSQL outperforms other classifiers in terms of number of parameters and test accuracy by reaching 99.52\% average test accuracy among 10 random experiments. \# Ps denotes \# Params.}
\label{table:result_mnist_binary}
\end{table}
% \linespread{1.0}

For multi-class classification, it is rarely discussed and tested in the literature of variational quantum classifiers. The one-vs.-all method is mentioned in \citeauthor{Schuld2018} but troublesome to implement. Therefore, we only compare the performance of VSQL with a single-layered classical neural network (NN). The experimental settings of the classical neural network are similar to VSQL, and it contains 7840 weights and 10 biases to map the 784-dimension input vectors to 10-dimensional output vectors. The results are summarized in Table \ref{table:result_mnist_10_class}. When using 9 different shadow circuits with each circuit depth $D=5$, VSQL could reach almost the same test accuracy with the single layer NN, but requiring much fewer parameters. Although this accuracy is not quite satisfied, it can still compete with the simplest classical NN. Notably, we find that if we select 1k samples (about 100 samples for each class) for training from 60k examples and choose the same size of test examples (10k), VSQL could achieve a higher test accuracy than NN (cf. the bottom half of Table \ref{table:result_mnist_10_class}). The above finding indicates that VSQL could extract high-level features from fewer training samples than NN, which may be a potential advantage of VSQL for future practical applications in the NISQ era.

%For the 10-class classification tasks, there are few other variational quantum algorithms can do this, therefore, we only compare VSQL with single layer neural network (NN). The experimental settings of the neural network are similar to VSQL and it has 7840 weights and 10 biases to map the 784-dimension input vectors to 10-dimension output vectors. The results are concluded in the top half of Table \ref{table:result_mnist_10_class}. When using 9 shadows with depth $D=5$, VSQL could reach almost the same test accuracy with the single layer NN, but requiring much less number of parameters. Although this accuracy is not quite satisfied, it can still compete with the simplest classical NN. What's more, an interesting experimental result is that if we just select 1k samples (about 100 samples for each class) for training from totally 60k examples and the test 10k examples are not changing, then VSQL could achieve a higher test accuracy than NN (cf. the bottom half of Table \ref{table:result_mnist_10_class}). This indicates that VSQL could extract high-level features from very fewer training samples than NN, which maybe a potential advantage of VSQL for future practical applications in the NISQ era.

% \linespread{1.5}
\begin{table}[t]
\small
%\caption{MNIST 10-class classification}
\centering
\begin{tabular}{c|c|c|c|c}
\toprule
 Methods & $n_s$ & $D$  & \# Ps & Test acc (\%)  \\\midrule
  NN (60k samples) &  /    & 1 & 7850  &  \textbf{91.63 $\pm$ 0.15} \\\midrule       
  VSQL (60k samples) & 5  & 5 & 520 &   87.69 $\pm$ 0.98      \\\midrule
  VSQL (60k samples) & 9  & 5 & 928 &   91.13 $\pm$ 0.51     \\\midrule
  &&&& \\\midrule
      NN (1k samples) & /  & 1 & 7850  &  86.36 $\pm$ 0.23 \\\midrule 
  VSQL (1k samples) & 5  & 5 & 520 &   83.92 $\pm$ 1.20      \\\midrule
  VSQL (1k samples) & 9  & 5 & 928 &   \textbf{87.39 $\pm$ 0.40}     \\ \toprule
\end{tabular}
% \linespread{1.0}

\caption{MNIST 10-class classification}
\label{table:result_mnist_10_class}
\end{table}

\section{Discussions}
\label{sec:conclusion}

We propose the VSQL framework, which adopts the idea of classical shadows to distinguish quantum data.  With theoretical justifications and numerical experiments, we have shown that VSQL for classification outperforms many other variational classifiers on the benchmark test of binary MNIST handwritten digit recognition with much fewer network parameters. In particular, in our framework, less noise will be introduced during the quantum-classical hybrid information processing as the number of quantum gates used is independent of the problem size. Simultaneously, by sampling a slice of the loss landscape, we briefly introduce the barren plateau problem and show the solution to escape from it. Notably, by adjusting the operating scope of shadow circuits, our approach can be easily implemented on the existing quantum devices with topological connectivity limits.

We believe that VSQL would open the possibility of many future directions. For example, it would be interesting to explore the applications on VSQL for generative models and unsupervised quantum machine learning tasks such as clustering. 
Furthermore, the online learning version of VSQL may also be a good future direction, see \cite{aaronson2018online,chen2020more,yang2020revisiting}.
We also expect that VSQL may shed light on other quantum applications on near-term quantum devices.

\section*{Acknowledgements}
We would like to thank Prof. Sanjiang Li, Prof. Yuan Feng and Youle Wang for helpful discussions.
G. L. acknowledges the support from the Baidu-UTS AI Meets Quantum project and the China Scholarship Council (No. 201806070139).  This work was partly supported by the Australian Research Council (Grant No: DP180100691).
This work was done when Z. S. was a research intern at Baidu Research.

\bibliography{ref}

% \begin{linenumbers}

% \appendix

\newpage

\onecolumn
\begin{center}
{\textbf{\Large Supplementary Material for ``VSQL: Variational Shadow Quantum Learning for Classification''}}
\end{center}
% \twocolumngrid

\renewcommand{\theequation}{S\arabic{equation}}
\renewcommand{\thethm}{S\arabic{thm}}
\renewcommand{\thefigure}{S\arabic{figure}}
\renewcommand{\thealgorithm}{S\arabic{algorithm}}
\setcounter{equation}{0}
\setcounter{figure}{0}
\setcounter{table}{0}
\setcounter{section}{0}
\setcounter{algorithm}{0}
\setcounter{thm}{0}

\vspace{2ex}
The outline of this supplementary material is as follows:
\vspace{1ex}

\begin{tabular}{l}
\vspace{1ex}
  \hspace{2em}  $\ast$ Algorithm of VSQL for binary classification: the inference process \\
  \vspace{1ex} \hspace{2em}  $\ast$    Further theoretical performance analysis about the number of repetitions of the shadow circuit \\
\vspace{1ex}    \hspace{2em} $\ast$ Further experimental results, including distinguishing noisy quantum states.\\
\vspace{1ex}    \hspace{2em} $\ast$  VSQL for multi-label classification \\
    \hspace{2em} $\ast$ Proof details for Theorems \ref{thm:classi_ability}, \ref{thm:n_1_n_2_lcoal_shadow} and \ref{thm:distinguish_psi_uv}, Proposition \ref{prop:barren_VSQL}.
\end{tabular}
\linespread{1.0}

\vspace{4ex}
\section{Algorithm of VSQL for binary classification: the inference process}

\begin{algorithm}[ht]
\caption{Variational shadow quantum learning (VSQL) for binary classification: the inference process}
 \label{alg:VSQL_binary_inference}
\begin{algorithmic}[1]
\REQUIRE The test data set $\mathcal{D}^{(test)}:=\{(\rho_{in}^{(m)}, y^{(m)}\in\{0,1\})\}_{m=1}^{N_{test}}$, the parameters $\bm{\theta}$, $\bm{w}$ and $b$ from the training process
\ENSURE The list of predicted labels and the test accuracy
\STATE Set the counter $n\_c = 0$, denoting the number of correct predicted labels
\FOR{$m = 1,\ldots, N_{test}$} 
\STATE Apply multi-times the shadow circuit $U(\bm{\theta})$ to the input density matrix $\rho_{in}^{(m)}$
\STATE Measure and estimate a series of expectations $\langle X\otimes X \rangle$, recorded as $o_i$'s
\STATE Feed these shadow features $o_i$'s into the classical neural network and obtain the output $\hat{y}^{(m)}\in [0,1]$
\IF{$\hat{y}^{(m)}\le 0.5$}
\STATE Set the predicted label as `0'
\ELSE 
\STATE Set the predicted label as `1'
\ENDIF
\IF{the predicted label == $y^{(m)}$}
\STATE $n\_c = n\_c + 1$
\ENDIF
\ENDFOR
\STATE Compute the test accuracy as $n\_c/N_{test}$
\end{algorithmic}
\end{algorithm}

\section{Further theoretical performance analysis }

\subsection{Number of repetitions for computing each shadow feature}

As we need to repeat the shadow circuits multiple times to estimate the shadow features, here we give the number of repetitions required in VSQL.
\begin{prop}\label{prop:num_repet}
Given a precision $\epsilon$, the number of repetitions of the shadow circuit for computing each shadow feature at error $\epsilon$, with probability at least $1-\eta$,  scales as $O\lr{\log({1}/{\eta})/{\epsilon^2}}$.
\end{prop}
This proposition is directly derived from the Chernoff–Hoeffding theorem \cite{Hoeffding1963}.
Furthermore, by utilizing these estimated shadow features, VSQL outputs the prediction value $\hat{y}$ and gives a label according to the following prediction rule
\begin{align}\label{eq:prediction_rule}
    \text{predicted label} = 
    \begin{cases}
    0, & \hat{y}<0.5 \\
    1, & \hat{y}\ge 0.5.
    \end{cases}
\end{align}
Therefore, in the inference process of VSQL, for an input state with the label $y \in \{0,1\}$,
if the predicted label is correct and the gap between the prediction value and 0.5 is $\tau$ under an infinite number of repetitions of the shadow circuits, then the actual number of repetitions, required to ensure that the input state is not misclassified, will be related to the gap $\tau$.

\begin{prop}\label{prop:num_repetitions_gap}
For an $n$-qubit quantum system, if we use $n_s$ shadow circuits and assume the final weights $w_i$ of the neural networks in VSQL are bounded as $|w_i|\le  C_w$, and the prediction gap is $\tau\in (0,0.5)$, then the actual number of repetitions for computing each shadow feature, with probability at least $1-\eta$, scales as  $O\lr{{n_s^2n^2C_w^2}\log({1}/{\eta})/{\tau^2}}$.
\end{prop}
\begin{proof}
If the estimated error of each shadow feature $o_i$ is $\delta$, then from Proposition \ref{prop:num_repet}, the number of repetitions, with probability at least $1-\eta$, is $O\lr{\log(1/\eta)/\delta^2}$.
What's more, due to
\begin{small}
\begin{align}
    \frac{\partial \hat{y}}{\partial o_i} \!:=\!\frac{\partial \sigma\lr{\sum_i w_i o_i +b}}{\partial o_i}\! =\! \hat{y}\lr{1\!-\!\hat{y}} \cdot w_i \le \frac{|w_i|}{4}  \le \frac{C_w}{4},
\end{align}
\end{small}
the error of $\hat{y}$ could be bounded as $\frac{1}{4}n_snC_w \delta$, where the first inequality follows from $0<\hat{y}<1$ and the term $n_sn$ means there are at most $n_sn$ shadow features. If we let $\frac{1}{4}n_snC_w \delta \le \tau$, the number of repetitions for computing each shadow feature is obtained.
\end{proof}

From Proposition \ref{prop:num_repetitions_gap}, we know, in the inference process of VSQL, if there is a large prediction gap, then VSQL will  require much less number of repetitions for computing each shadow feature to ensure obtaining a correct predicted label.

\section{Further experimental results}

\subsection{Background of QSD}
Quantum state discrimination (QSD) is fundamental to the theory of quantum cryptography \cite{Bennett1992} and quantum communications \cite{Nielsen2002,Barnett2009,Bae2015}. It is usually defined as follows: can we recognize a quantum state $\rho_k$ from a set of quantum states $\{\rho_i\}_{i=1}^{N}$ with known probability distribution $\{q_i\}_{i=1}^{N}$ for the quantum system to be in each corresponding state, via certain measurements? This is non-trivial since arbitrary pre-measurement manipulations and measurement does not always extract useful classical information from the quantum system. Although, in principle, an optimal projective measurement can be designed according to the Helstrom bound \cite{Helstrom1969} by minimizing the average guessing error, this kind of strategy is difficult to find in general and the optimal strategy is only know for limited cases. Furthermore, even if we could obtain this optimal measure, the amount of information that we can extract is still limited by the Holevo bound \cite{Holevo1973} and the physical realization of those measure remain challenging given the hardware restrictions mentioned before. From our perspective, it is natural to think of the combination of variationally searching appropriate pre-measurement manipulations and hardware-efficient measure instead of directly finding the optimal measure. In particular, those locally-operated shadow circuits $U(\bm{\theta})$ will function as the pre-measurement manipulation in VSQL and the Pauli measure $X$ on each qubit is indeed hardware-efficient.

\subsection{Classification of multi-class quantum states}

As declared before, our method can be easily extended to multi-class classification and numerically verified. Here, we take three different categories as an example. The data set we choose is again taken from Eq. \eqref{eq:quantum_data_set}, but adding the following third family $\rho_3(t) = \op{\psi_t}{\psi_t}$ with a new probability distribution $\{q_i\} = \{\frac{1}{4}, \frac{1}{2}, \frac{1}{4} \}$,
\begin{align}
    \ket{\psi_t} &= \Lr{\sqrt{1-t^2}, t, 0, 0}^\top,  
\end{align}
where $t\in [0,1]$. We shall note that these three families of states $\ket{\psi_u},\ket{\psi_v}$ and $\ket{\psi_t}$ are mutually non-orthogonal unless $u,v,t$ are taken as 0 or 1. Hence, it's difficult to distinguish them via POVM method \cite{Bae2015}. 
Now we use VSQL to distinguish them. Similarly, we generate another 100 density matrices $\rho_3$ which are labeled as `[0,0,1]' (here we use one-hot vectors to denote the labels, i.e., `[1,0,0]' for $\rho_1$ and `[0,1,0]' for $\rho_2$). The other experimental settings are identical to the binary case except for the softmax activation function used in FCNN. Similar learning curves of the training process for the loss and the validation accuracy are demonstrated in Fig. \ref{fig:qsd_multi10}, which shows VSQL could perfectly distinguish multi-class quantum states by reaching 100\% validation accuracy. The fluctuations on the loss curve are probably due to the design of cross-softmax loss function and the existence of the highly non-orthogonal data samples. As a consequence, the validation accuracy is also jiggling around but gradually converge to the theoretical maximum. Similar to the binary case, we repeat the simulation by sampling $u,v,t \in \text{Uni}\Lr{0.1,0.9}$ and summarize the results in Fig. \ref{fig:qsd_multi09}. This eliminates the extreme cases $u,v,t \in \{0,1\}$ which reduce the multi-class to binary classification. As expected, smaller fluctuations are observed which means our method could unambiguously distinguish multi-class quantum states.

\subsection{Distinguishing noisy quantum states}

In practice, it is inevitable to deal with noise on the current quantum hardware which leads to noisy quantum states. Thus, it is essential to verify whether VSQL could distinguish noisy quantum states if we want to realize VSQL on the hardware in near future. In this subsection, we will run simulations on a pair of constructed noisy quantum states with a high fidelity.

%In practice, the quantum states are probability generated accompanied by noise, hence, it's essential to verify whether VSQL could distinguish noisy quantum states. In this subsection, we will run simulations on a pair of constructed noisy quantum states with a high fidelity.

 \begin{figure}[ht]
	\centering
	\subfigure[Noise probability equals to 0.1]{\label{fig:visual01}
		\includegraphics[width=0.48\textwidth]{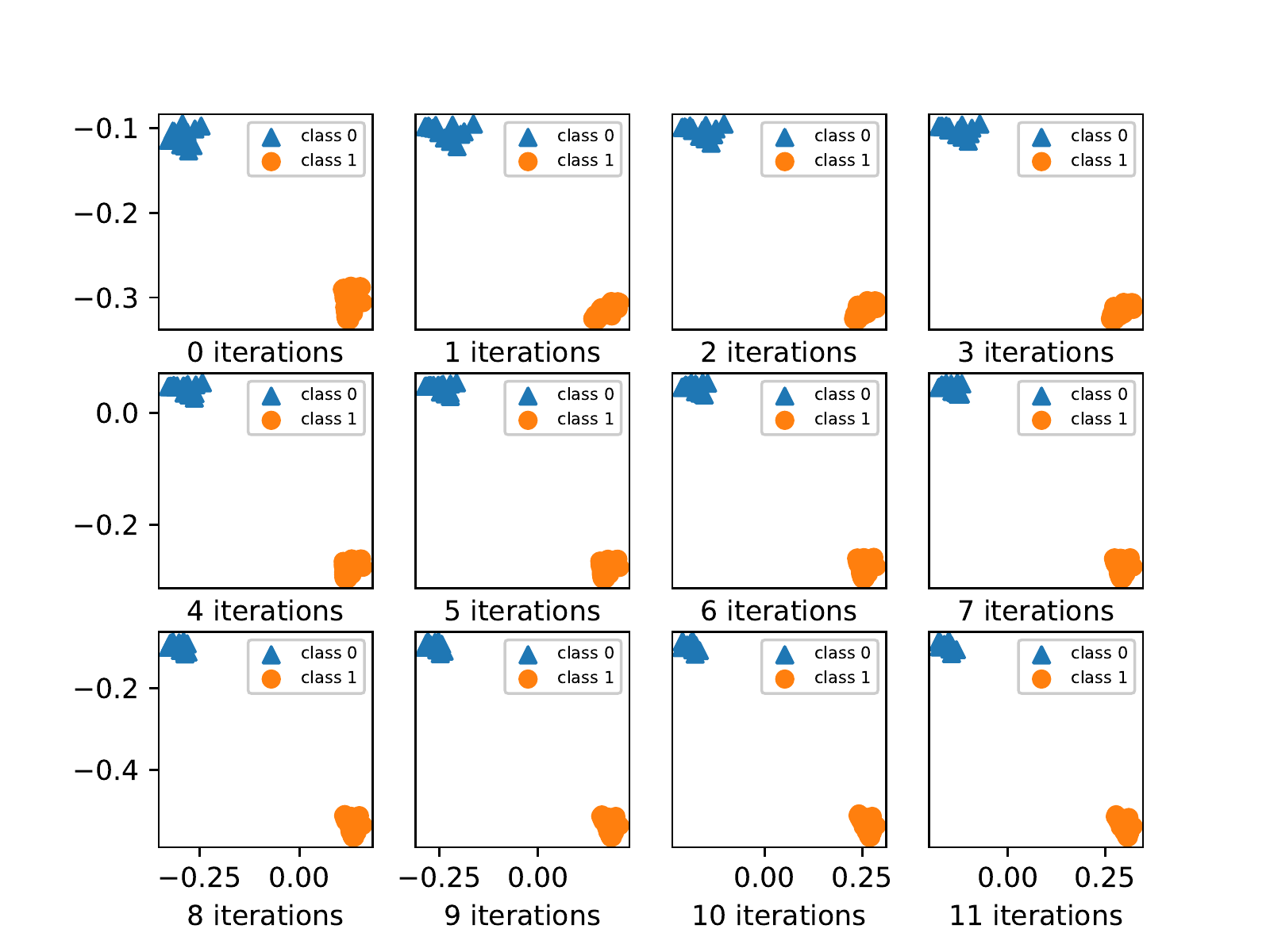}}
	\subfigure[Noise probability equals to 0.5]{\label{fig:visual05}
		\includegraphics[width=0.48\textwidth]{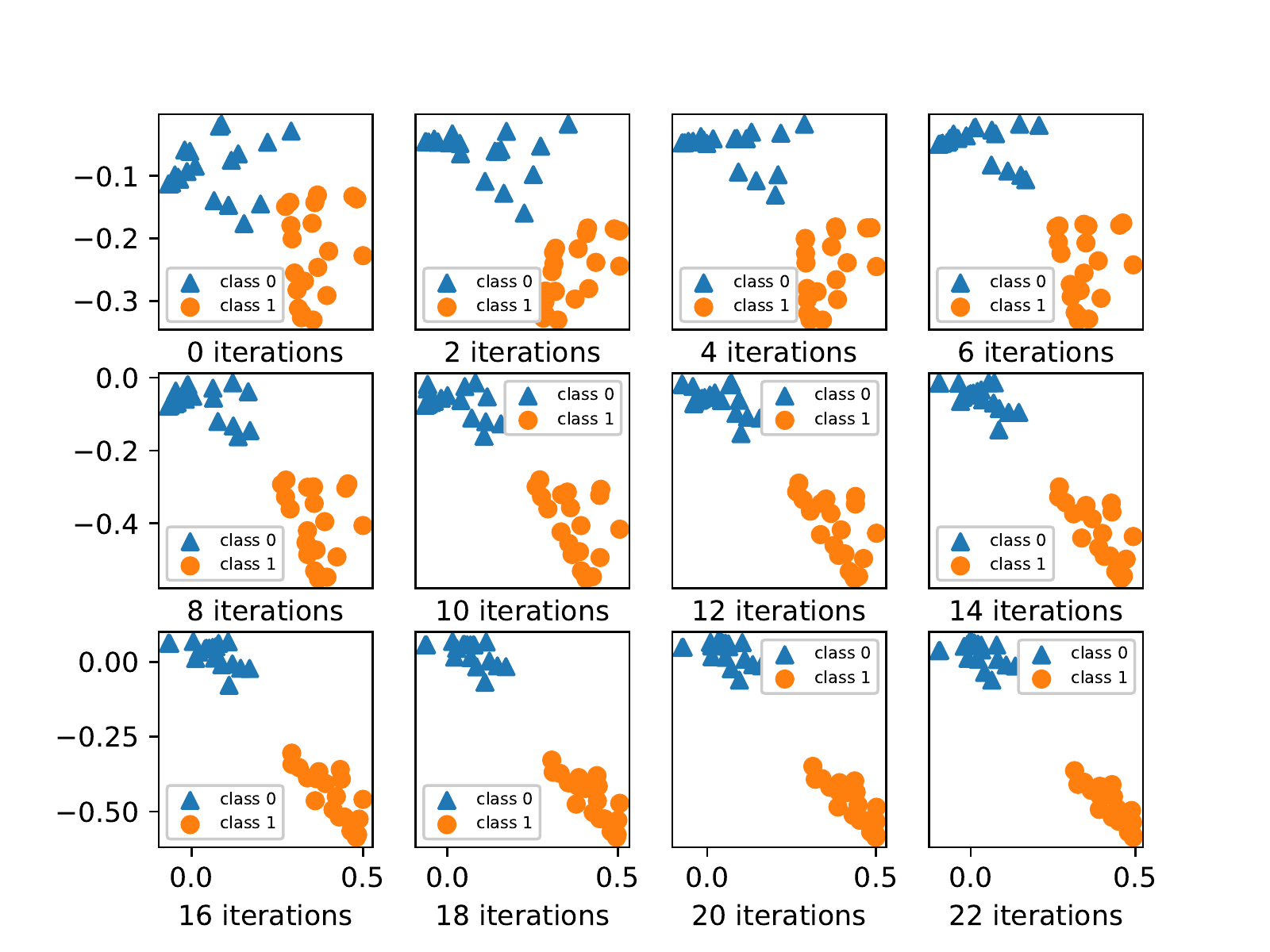}}
	\subfigure[Noise probability equals to 0.9]{\label{fig:visual09}
		\includegraphics[width=0.48\textwidth]{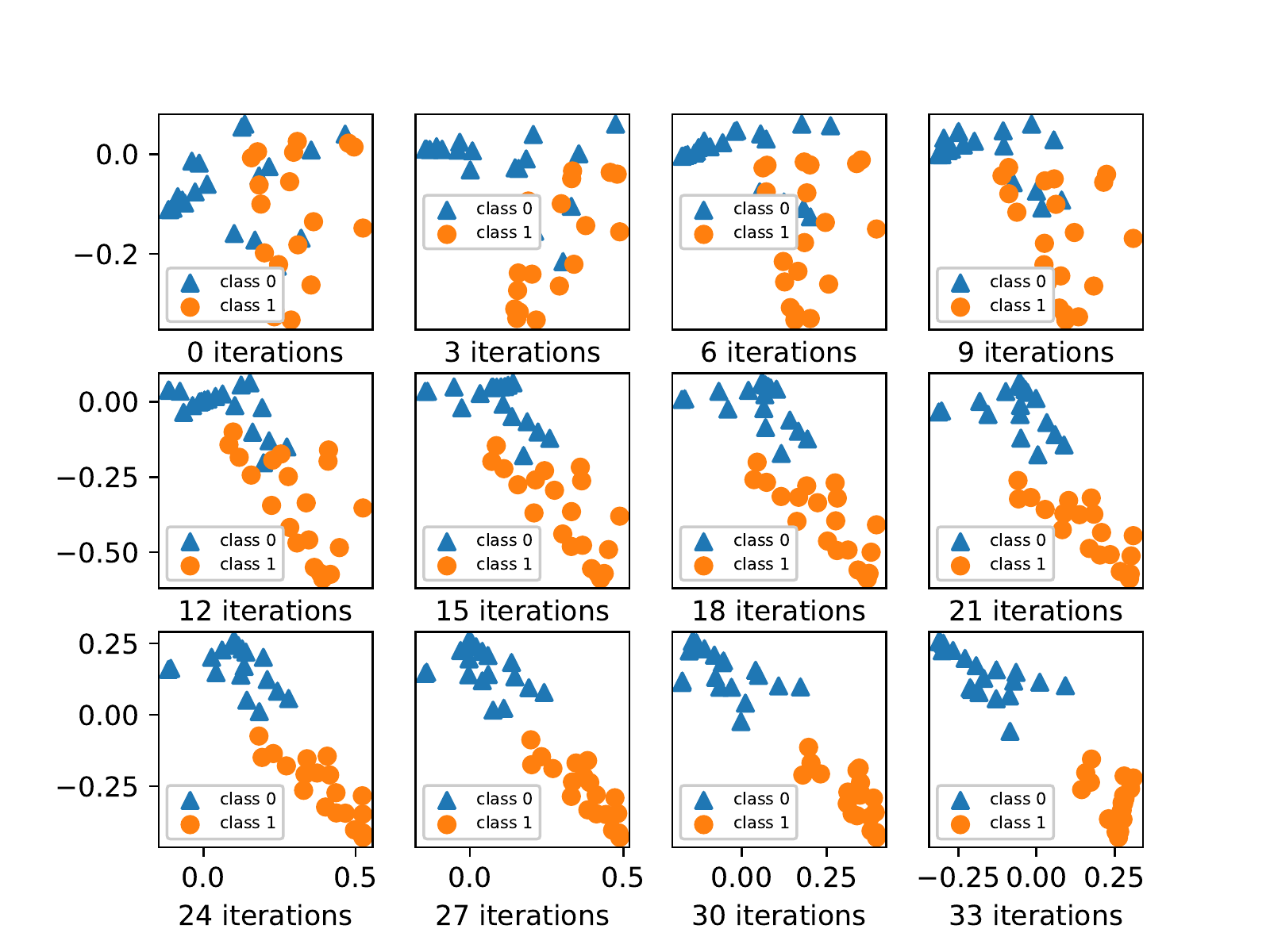}}
			\subfigure[Test accuracy curves]{\label{fig:visual_test_acc}
		\includegraphics[width=0.45\textwidth]{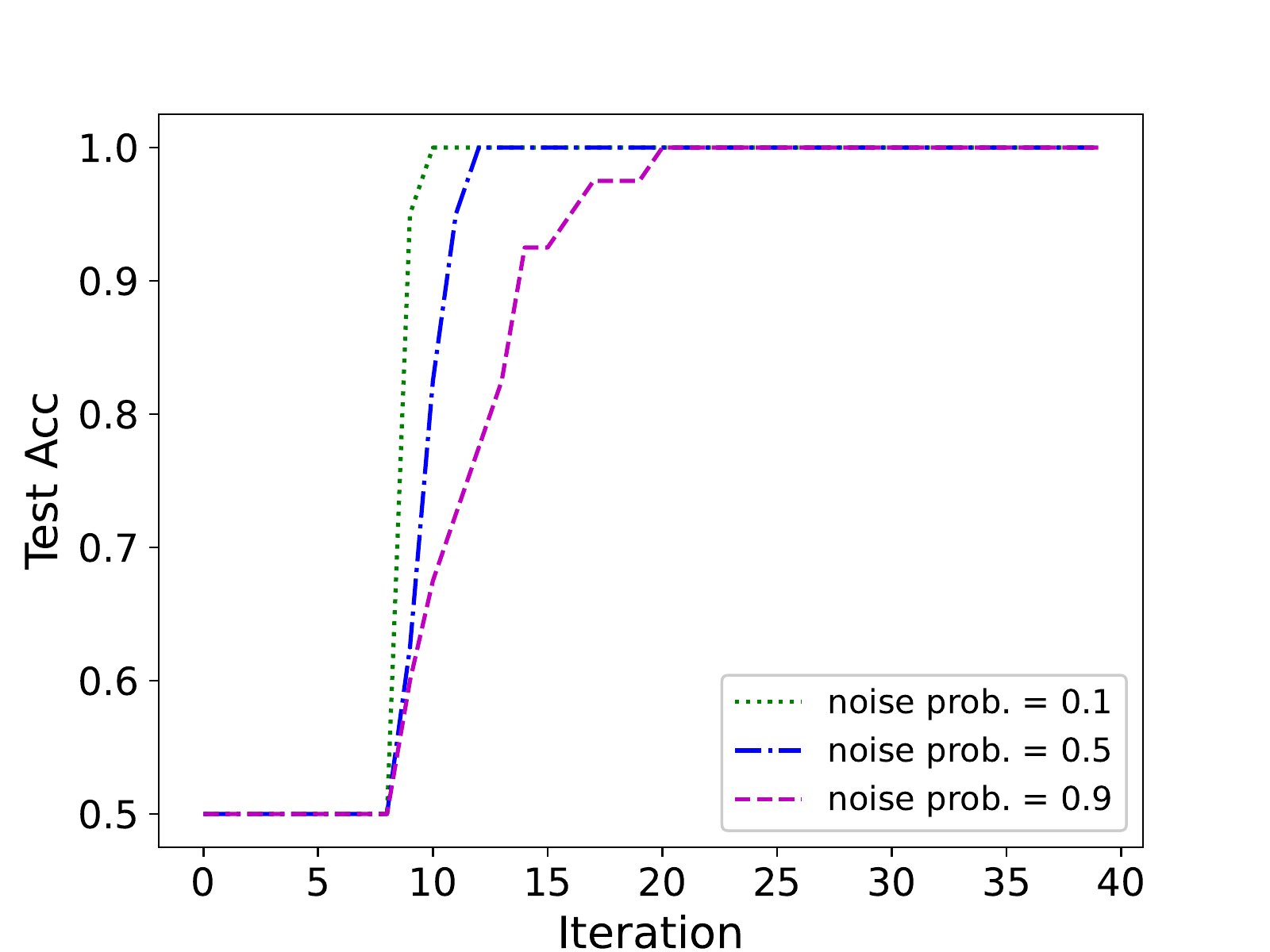}}
	\caption{Shadow features changing with the number of iterations under different noise probabilities, and the corresponding test accuracy curves.}
	\label{fig:visualize}
\end{figure}

The procedure for generating this pair of simulated quantum states is as follows:
(i) We first construct two pure states (in 3 qubits as an example) with high fidelity: $ \rho^{(0)} = \op{\psi_0}{\psi_0}$ and $ \rho^{(1)} = \op{\psi_1}{\psi_1}$, which are labeled 0 and 1 respectively, where $ \ket{\psi_0} ={1}/{2}\lr{\ket{000}+\ket{001}+\ket{010}+\ket{011}}$,
    $\ket{\psi_1} ={1}/{\sqrt{3}}\lr{\ket{000}+\ket{001}+\ket{010}}$.
(ii) Then we sample a unitary $U_s$ from matrix QR decomposition and apply it to these two pure states: $  U_s\rho^{(0)} U_s^\dag$ and $U_s\rho^{(1)} U_s^\dag$.
(iii) Last we imposed a Pauli noise on the states, i.e., 
\begin{align}
    \rho^{(i)}_{in}= (1-p_i) U_s\rho^{(i)} U_s^\dag +\frac{p_i}{3} \sum_{j=1}^3 E_j U_s\rho^{(i)} U_s^\dag E_j^\dag,
\end{align} 
where $i\in \{0,1\}$, $E_1=P\otimes \mathbb I\otimes \mathbb I, E_2= \mathbb I \otimes P\otimes \mathbb I, E_3= \mathbb I\otimes \mathbb I\otimes P$, $P\in\{X,Y,Z\}$, and
the noise probability $p_i$ is sampled from a 
uniform distribution $[0, \mathcal{P}]$ with a constant  $ \mathcal{P}$  between 0 and 1.

% \begin{enumerate}
%     \item We first construct two pure states (in 3 qubits as an example) with high fidelity: $ \rho^{(0)} = \op{\psi_0}{\psi_0}$ and $ \rho^{(1)} = \op{\psi_1}{\psi_1}$, which are labeled 0 and 1 respectively, where
% \begin{align}
%     \ket{\psi_0} &=\frac{1}{2}\lr{\ket{0}+\ket{1}+\ket{2}+\ket{3}}, \\
%     \ket{\psi_1}& =\frac{1}{\sqrt{3}}\lr{\ket{0}+\ket{1}+\ket{2}}.
% \end{align}
% \item Then we sample a unitary $U_s$ from matrix QR decomposition and apply it to these two pure states: 
% \begin{align}
%   U_s\rho^{(0)} U_s^\dag\quad \text{and} \quad U_s\rho^{(1)} U_s^\dag. 
% \end{align}
% \item Last we imposed a noise on the states, i.e., 
% \begin{align}
%     \rho^{(i)}_{in}= (1-p_i) U_s\rho^{(i)} U_s^\dag +\frac{p_i}{3} \sum_{j=1}^3 E_j U_s\rho^{(i)} U_s^\dag E_j^\dag,
% \end{align} 
% where $i\in \{0,1\}$, $E_1=P\otimes \mathbb I\otimes \mathbb I, E_2= \mathbb I \otimes P\otimes \mathbb I, E_3= \mathbb I\otimes \mathbb I\otimes P$, $P\in\{X,Y,Z\}$, and
% the noise probability $p_i$ is sampled from a 
% uniform distribution $[0, \mathcal{P}]$ with a constant  $ \mathcal{P}$  between 0 and 1.
% \end{enumerate}

In our experiment, given a fixed  $ \mathcal{P}$ , we sample 40 probabilities $p_0$'s and 40 $p_1$'s and thus generating 40 noisy quantum states for each class. Amongst these states, 50\% is for training and the remaining 50\% for testing. We employ one 2-local shadow circuit, which is similar to the Fig. \ref{fig:ansatz_mnist_binary} with depth $D=1$. The learning rate is set to 0.1 and the other experimental settings are similar to the above two experiments.

In order to explore the sensitivity of VSQL to the noise level, we conduct multiple experiments by setting $ \mathcal{P}$  as $ 0.1, 0.5, 0.9$, respectively. The test accuracy curves in the training process is illustrated in Fig. \ref{fig:visual_test_acc}, where we see intuitively that all the test accuracy could reach 100\% after 20 iterations also, even though given a higher noise level. It also shows that the lower the noise level is, the faster the test accuracy increases, which is in line with our intuition. 
Furthermore, for the sake of figuratively understanding the classification ability of VSQL, we record the two shadow features (in Eq. \eqref{eq:def_observ}) of the 40 test quantum states in each training iteration. The results with different noise probabilities are illustrated in Figs. \ref{fig:visual01}, \ref{fig:visual05} and \ref{fig:visual09}, respectively. We observe that it is easier to distinguish when the noise probability equals to 0.1 or 0.5, as the corresponding two classes of points are distributed in two clusters initially. However, even the two classes of points are interlaced initially when the noise probability equals to 0.9, they will be gradually separated into two clusters with the training process going on.

\section{Variational shadow quantum learning for multi-label classification}
\label{appendix:sec:multi_class}

In this section, we will simply describe the VSQL for multi-label classification, which consists of overall sketch, loss function, analytical gradients, number of parameters, number of repetitions and theoretical classification ability. Most of the settings are the same as the binary case, except for the final activation function, i.e., sigmoid activation function for binary case and softmax activation function for multi-label case.

\subsection{Sketch of VSQL for multi-label classification}

In this subsection, we will supplement the sketch of VSQL for multi-label classification, see Fig. \ref{fig:scheme_vsqc_multi}. And the corresponding training and inference processes are concluded in Algorithms \ref{alg:VSQL_multi_train} and \ref{alg:VSQL_multi_inference}, respectively.

\begin{figure}[ht]
	\centering
		\includegraphics[width=0.8\textwidth]{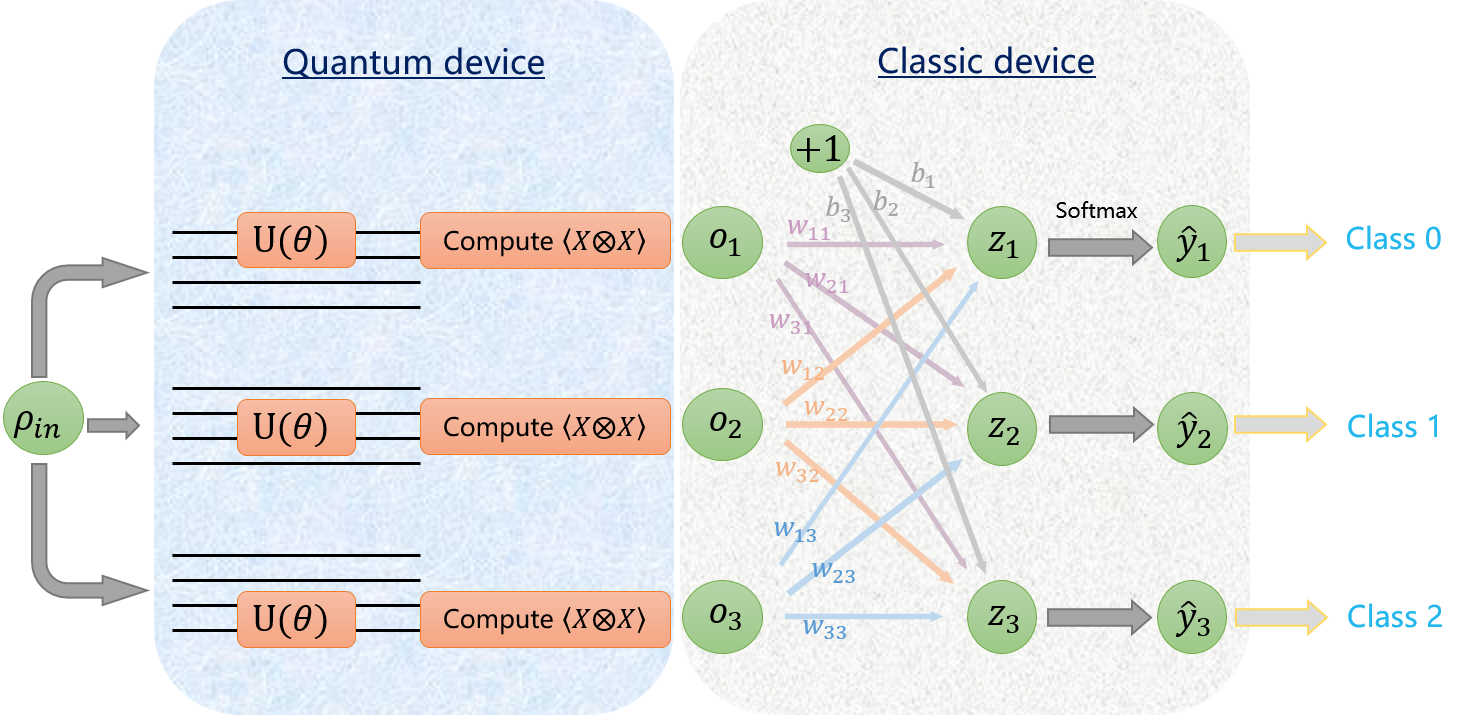}
	\caption{Sketch of variational shadow quantum learning (VSQL) for multi-label classification with $n=4$, $n_{qsc} =2$ and $K=3$. In the quantum device, the shadow circuit is implemented on the subspace of input state $\rho_{in}$. Sliding through the whole system to collect the Pauli-$(X\otimes X)$ expectations, i.e., shadow features. In the classic device, the resulting shadow features $o_i$'s are fed into a fully-connected neural network (FCNN). Here, the softmax activation function is employed and the output $\hat{y}$ is a $K$-dimensional vector for multi-label case.}
	\label{fig:scheme_vsqc_multi}
\end{figure}

\begin{algorithm}[ht]
\caption{VSQL for multi-label classification: the training process}
 \label{alg:VSQL_multi_train}
\begin{algorithmic}[1]
\REQUIRE The training data set $\mathcal{D}^{(train)}:=\{(\rho_{in}^{(m)}, y^{(m)}\in \mathbb R^K)\}_{m=1}^{N_{train}}$, $EPOCH$, optimization procedure
\ENSURE The final parameters $\bm{\theta}^*$, $\bm{W}^*$ and $\bm{b}^*$, and the list of losses 
\STATE Initialize the parameters $\bm{\theta}$ of the 2-local (for example) shadow circuit $U(\bm{\theta})$ from uniform distribution Uni$[0, 2\pi]$ and $\bm{W},\bm{b}$ from Gaussian distribution $\mathcal{N}(\bm{0},\mathbb I)$
\FOR{$ep=1,\ldots, EPOCH$}
    \FOR{$m = 1,\ldots, N_{train}$} 
    \STATE Apply multi-times the shadow circuit $U(\bm{\theta})$ to the input density matrix $\rho_{in}^{(m)}$
    \STATE Measure the subsystem and estimate a series of expectations $\langle X\otimes X \rangle$, recorded as $o_i$'s
    \STATE Feed the shadow features $o_i$'s into the classical neural network and obtain the output $\hat{y}^{(m)}$
    \STATE Compute the accumulated loss $\sum_{k=1}^{K} y^{(m)}_k \log  \hat{y}^{(m)}_k$ and update accordingly the parameters $\bm{\theta}$, $\bm{W}$ and $\bm{b}$ via gradient-based optimization procedure
    \ENDFOR
    \IF{the stopping criterion is satisfied}
    \STATE Break
    \ENDIF
\ENDFOR
\end{algorithmic}
\end{algorithm}

\begin{algorithm}[t]
\caption{VSQL for multi-label classification: the inference process}
 \label{alg:VSQL_multi_inference}
\begin{algorithmic}[1]
\REQUIRE The test data set $\mathcal{D}^{(test)}:=\{(\rho_{in}^{(m)}, y^{(m)}\in \mathbb R^K)\}_{m=1}^{N_{test}}$, the parameters $\bm{\theta}$, $\bm{W}$ and $\bm{b}$ from the training process
\ENSURE The list of predicted labels and the test accuracy
\STATE Set the counter $n\_c = 0$, denoting the number of correct predicted labels
\FOR{$m = 1,\ldots, N_{test}$} 
\STATE Apply multi-times the shadow circuit $U(\bm{\theta})$ to the input density matrix $\rho_{in}^{(m)}$
\STATE Measure and estimate a series of expectations $\langle X\otimes X \rangle$, recorded as $o_i$'s
\STATE Feed these shadow features $o_i$'s into the classical neural network and obtain the output $\hat{y}^{(m)}\in \mathbb R^K$
\IF{$ l = \mathop{\text{argmax}}\limits_{k} \{\hat{y}^{(m)}_k\}$}
\STATE Set the predicted label as `$l-1$'
\ENDIF
\IF{$ \mathop{\text{argmax}}\limits_{k} \{\hat{y}^{(m)}_k\}$ == $ \mathop{\text{argmax}}\limits_{k} \{{y}^{(m)}_k\}$ }
\STATE $n\_c = n\_c + 1$
\ENDIF
\ENDFOR
\STATE Compute the test accuracy as $n\_c/N_{test}$
\end{algorithmic}
\end{algorithm}

\subsection{Loss function}
Given the data set $\mathcal{D}:=\{(\rho_{in}^{(m)}, y^{(m)})\}_{m=1}^N \subset \mathbb C^{2^n\times 2^n}\times \mathbb R^K$ and $n_{qsc}$-local shadow circuits, where
$y^{(m)}$ is a one-hot vector which indicates the category to which the $m^{\text{th}}$ data sample $\rho_{in}^{(m)}$ belongs. For example, if $K=3$, $y^{(m)}=[1,0,0]^\top$ indicates the $m^{\text{th}}$ sample belongs to class 0, $y^{(m)}=[0,1,0]^\top$ for class 1 and $y^{(m)}=[0,0,1]^\top$ for class 2.
The loss function of VSQL for multi-label classification is derived from cross-entropy \cite{Goodfellow-et-al-2016}:
\begin{align}\label{eq:def_loss_multi}
    \mathcal{L}(\bm{\theta},\bm{W},\bm{b};\mathcal{D}) := - \frac{1}{N}\sum_{m=1}^N \sum_{k=1}^{K} y^{(m)}_k \log  \hat{y}^{(m)}_k\lr{\rho_{in}^{(m)};\bm{\theta},\bm{W},\bm{b}}.
    \end{align}
Here, the output $K$-dimensional vector $\hat{y}^{(m)}$ of VSQL is defined as follows:
\begin{align}
    \hat{y}^{(m)}\lr{\rho_{in}^{(m)};\bm{\theta},\bm{W},\bm{b}} := \sigma\lr{\sum_{i=1}^{n-n_{qsc}+1} \bm{w}_i o^{(m)}_i\lr{\rho_{in}^{(m)};\bm{\theta}}+ \bm{b}},
\end{align}
where $\bm{W}=[\bm{w}_1,\bm{w}_2,\ldots,\bm{w}_{n-n_{qsc}+1}] \in \mathbb R^{K\times (n-n_{qsc}+1)}$, $\bm{b}\in \mathbb R^{K\times 1}$,
$\sigma(\bm{z})= \frac{\mathrm{e}^{\bm{z}} }{\sum_j\mathrm{e}^{z_j}}$ denotes
the softmax activation function and the shadow features $o_i$ are calculated through
\begin{align}\label{eq:def_observ_multi}
    o^{(m)}_i\lr{\rho_{in}^{(m)};\bm{\theta}} = \Tr \lr{ \rho_{in}^{(m)}  ( \mathbb I\otimes \cdots \otimes U^\dag(\bm{\theta}) O U(\bm{\theta})\otimes \cdots \otimes \mathbb I)  }.
\end{align}
Note that the shadow circuit $U(\bm{\theta})$ and the Hermitian operator $O=X\otimes\cdots\otimes X $ are applied on the same local qubits.
We also note that the calculation of these shadow features and the construction of shadow circuits are the same as the binary case, i.e.,
% What's more, $U(\bm{\theta})$ is usually decomposed as a chain of unitary operators:
\begin{align}\label{eq:def_U_multi}
    U(\bm{\theta})=\prod_{l=L}^1 U_l(\theta_l)V_l,
\end{align}
where $ U_l(\theta_l)=  \exp(-i\theta_l/2P_l)$ with the Pauli product operator $P_l$ and $V_l$ denotes a fixed operator such as Identity, CNOT and so on.

\subsection{Analytical gradients}
For each data sample $(\rho_{in}^{(m)}, y^{(m)})$ in the data set $\mathcal{D}$ and assume $y^{(m)}_k=1$, the partial derivatives with respect to the parameters $w_{ji},b_j$ and $\theta_l$ are calculated as follows:
\begin{align}
    \frac{\partial \mathcal{L}(\bm{\theta},\bm{W},\bm{b};\rho_{in}^{(m)}, y^{(m)})}{\partial w_{ji}} &= \begin{cases}
    \lr{\hat{y}^{(m)}_k-1} \cdot o^{(m)}_i, & j= k \\
     \hat{y}^{(m)}_j \cdot o^{(m)}_i, & j\neq k
    \end{cases}   \label{eq:derivative_w_i_multi}        \\  
   \frac{\partial \mathcal{L}(\bm{\theta},\bm{W},\bm{b};\rho_{in}^{(m)}, y^{(m)})}{\partial b_{j}} &= \begin{cases}
    \lr{\hat{y}^{(m)}_k-1}, & j= k \\
     \hat{y}^{(m)}_j, & j\neq k
    \end{cases}   \label{eq:derivative_b_multi}  \\
% \end{align}
% \begin{align}
    \frac{\partial \mathcal{L}(\bm{\theta},\bm{W},\bm{b}; \rho_{in}^{(m)}, y^{(m)})}{\partial \theta_l} & = \sum_{i=1}^{n-n_{qsc}+1}  \sum_{j=1}^{K} \lr{\hat{y}^{(m)}_j  w_{ji} -w_{ki} } \frac{\partial o^{(m)}_i\lr{\bm{\theta};\rho_{in}^{(m)}}}{\partial \theta_l},
     \label{eq:derivative_theta_multi}
\end{align}
where $\hat{y}^{(m)}$ and $o^{(m)}_i$ are the corresponding abbreviations.
It should be noted that the last term in Eq. \eqref{eq:derivative_theta_multi}, i.e., the partial gradient $ {\partial o^{(m)}_i}/{\partial \theta_l}$, is resolved the same as Eq. \eqref{eq:ansatz_derivative_theta}.

\subsection{Number of parameters}

The number of parameters of VSQL for multi-label classification is summarized in the following proposition.
\begin{prop}\label{prop:num_params_multi}
For an $n$-qubit quantum system, if we use $n_s$ shadow circuits, then the number of parameters of VSQL for $K$-label classification is 
\begin{align}
     \text{\# Params} =  \text{\# Params} \big|_{\text{ in shadow circuits}} + \text{\# Params} \big|_{\text{ in NN}} = n_s n_{qsc} D +\Lr{ n_s \lr{n-n_{qsc}+1}+1}K,
\end{align}
where we denote by $n_{qsc}$ the \textbf{n}umber of \textbf{q}ubits of the \textbf{s}hadow \textbf{c}ircuits and assume each shadow circuit consists of $D$ layers with $n_{qsc}$ parameters in each layer. 
\end{prop}

\subsection{Number of repetitions for computing each shadow feature}

Since the number of repetitions to estimate the shadow features is the same as the binary case, here we merely rewrite it simply.
  
\renewcommand{\thethm}{\ref{prop:num_repet}}
\begin{prop}
Given a precision $\epsilon$, the number of repetitions of the shadow circuit for computing each shadow feature at error $\epsilon$, with probability at least $1-\eta$,  scales as $O\lr{\log({1}/{\eta})/{\epsilon^2}}$.
\end{prop}
\renewcommand{\thethm}{S\arabic{thm}}
\addtocounter{thm}{-1}

Furthermore, by utilizing these estimated shadow features, VSQL outputs the prediction vector $\hat{y}$ and gives a label according to the following prediction rule
\begin{align}\label{eq:prediction_rule_multi}
    \text{predicted label} = \mathop{\text{argmax}}\limits_{k} \{\hat{y}_k\} -1.
\end{align}
Therefore, in the inference process of VSQL for multi-label classification, for an input state with the label $y \in \mathbb R^K$,
if the predicted label is correct and the gap between the largest two values of the prediction vector $\hat{y}$ is $\tau$ under an infinite number of repetitions of the shadow circuits, then the actual number of repetitions, required to ensure that the input state is still correctly classified, will be similarly related to the gap $\tau$.
And an analogous result is concluded in Proposition \ref{prop:num_repetitions_gap_multi}.

\begin{prop}\label{prop:num_repetitions_gap_multi}
For an $n$-qubit quantum system, if we use $n_s$ shadow circuits and assume the final weights $w_{ji}$ of the neural networks in VSQL are bounded as $|w_{ji}|\le  C_w$ for all $i,j$, and the prediction gap is $\tau\in (0,1)$, then the actual number of repetitions for computing each shadow feature, with probability at least $1-\eta$, scales as $O\lr{{n_s^2n^2C_w^2}\log({1}/{\eta})/{\tau^2}}$.
\end{prop}
\begin{proof}
If the estimated error of each shadow feature $o_i$ is $\delta$, then from Proposition \ref{prop:num_repet}, the number of repetitions, with probability at least $1-\eta$, is $O\lr{\log(1/\eta)/\delta^2}$.
What's more, due to
\begin{align}
    \frac{\partial \hat{y}_j}{\partial o_i}
    & :=\frac{\partial }{\partial o_i}\lr{\frac{\text{e}^{\sum_i w_{ji}o_i +b_j}}{ \sum_l \text{e}^{\sum_i w_{li}o_i +b_l}} }
     = \hat{y}_j\lr{1-\hat{y}_j} \cdot w_{ji} + \sum_{l,l\neq j} -\hat{y}_j\hat{y}_l \cdot w_{li}
    \\
    & \le \hat{y}_j C_w \lr{1-\hat{y}_j+\sum_{l,l\neq j} \hat{y}_l} = 2\hat{y}_j C_w \lr{1-\hat{y}_j} \le \frac{1}{2}C_w,
\end{align}
for all $j=1,\ldots,K$,
the error of each value of $\hat{y}$ could be bounded as $\frac{1}{2}n_snC_w \delta$. If we let $\frac{1}{2}n_snC_w \delta \le \frac{\tau}{2}$, the number of repetitions for computing each shadow feature is obtained.
\end{proof}

\subsection{Theoretical classification ability}
Since it is too complex to explore the theoretical classification ability of VSQL for multi-label classification, here, we merely give an sufficient condition which is concluded in  Theorem \ref{thm:classi_ability_multi}.
% is similar to the binary case and the corresponding necessary and sufficient conditions are concluded as follows.
From Theorem \ref{thm:classi_ability_multi}, we could also directly induce the corresponding results analogous to Corollary \ref{coro:classi_ability_partial_trace} and Theorem \ref{thm:n_1_n_2_lcoal_shadow}, here we omit them.

\begin{thm}\label{thm:classi_ability_multi}
Given $K$ types of input density matrices $\rho_{in}^{(0)}$, $\rho_{in}^{(1)}$ and up to $\rho_{in}^{(K-1)}$ with labels 0, 1 up to $K-1$, respectively. If there exists a group of $\bm{\theta} $ that makes at least one group of shadow features $o_i^{(0)}$, $o_i^{(1)}$, $\ldots, o_i^{(K-1)}$ different, i.e., $| o_i^{(k)}-o_i^{(k^\prime)}|  >0 $ for all $k\neq k^\prime$, then VSQL is theoretically capable of distinguishing them.
\end{thm}
\begin{proof}
Without loss of generality, we assume $i=1$ and $o_1^{(0)}<o_1^{(1)}<\cdots<o_1^{(K-1)}$. By simply setting $w_{ji}=0$ for $i\neq 1,j=1,2,\ldots,K$, we have 
\begin{align}
z_1=w_{11} o_1^{(k)}+b_1,\quad  z_2=w_{21} o_1^{(k)}+b_2,\quad \cdots\quad z_K=w_{K1} o_1^{(k)}+b_K.
\end{align}
Our goal is to prove that $z_{k+1}$ is the largest one for any $o_1^{(k)}$, $k=0,1,\ldots,K-1$, via adjusting $\bm{w}_1$ and $\bm{b}$.

Now if we define $o_1^{(-1)}=o_1^{(0)}-1$, set $b_j= -w_{j1}o_1^{(j-2)}$ for all $j=1,2,\ldots,K$, and set $0<w_{11}<\cdots<w_{j1}<\cdots <w_{K1}$ such that
\begin{align}\label{eq:w_neq_multi}
    w_{j1} > w_{j-1,1}\frac{o_1^{(j-1)}-o_1^{(j-3)}}{o_1^{(j-1)}-o_1^{(j-2)}}, \qquad \text{for}\quad  j>1.
\end{align}
Then we could easily verify that for any $o_1^{(k)}$,
\begin{enumerate}
    \item if $l\ge k+2$, then 
\begin{align}
z_l=w_{l1}o_1^{(k)} + b_l =w_{l1}o_1^{(k)} -w_{l1} o_1^{(l-2)} =w_{l1} \lr{o_1^{(k)} -o_1^{(l-2)}}\le 0;
\end{align}
\item  if $l= k+1$, then
\begin{align}
    z_l=w_{l1}o_1^{(k)} + b_l =w_{l1}o_1^{(k)} -w_{l1} o_1^{(l-2)} = w_{l1} \lr{o_1^{(k)} -o_1^{(k-1)}}> 0;
\end{align}
\item from Eq. \eqref{eq:w_neq_multi}, we have
\begin{align}
    z_{k+1} =w_{k+1,1} \lr{o_1^{(k)} -o_1^{(k-1)}} > w_{k1} \lr{o_1^{(k)} -o_1^{(k-2)}}> w_{k1} \lr{o_1^{(k-1)} -o_1^{(k-2)}} = z_k,
\end{align}
and go on we have $z_{k+1}>z_{k}>z_{k-1}>\cdots > z_1= w_{11} \lr{o_1^{(0)} -o_1^{(-1)}}=w_{11}>0$.
\end{enumerate}
Based on the above three cases, we obtain that $z_{k+1}$ is the largest one, i.e., $\hat{y}_{k+1}$ is largest. That is to say, for any input density matrix $\rho_{in}^{(k)}$, VSQL outputs the predicted label $=  \mathop{\text{argmax}}\limits_{k} \{\hat{y}_k\} -1 = k$, which means classifying correctly.
\end{proof}

\section{Proof details}

\subsection{Proof of Theorem \ref{thm:classi_ability}}
\label{appedix:thm:classi_ability}

\renewcommand\thethm{\ref{thm:classi_ability}}
\begin{thm}
Given two types of input density matrices $\rho_{in}^{(0)}$ and $\rho_{in}^{(1)}$ with labels 0 and 1, respectively, if there exists a group of $\bm{\theta} $ that makes at least one pair of shadow features $o_i^{(0)}$ and $o_i^{(1)}$ different, i.e., $| o_i^{(0)}-o_i^{(1)}|  >0 $, then VSQL can distinguish them, vice versa.
\end{thm}
\renewcommand{\thethm}{S\arabic{thm}}
\addtocounter{thm}{-1}

\begin{proof}
Sufficiency:
Without loss of generality, we assume $i=1$ and $o_1^{(0)}\!<\!o_1^{(1)}$. By simply setting $w_1\!=\!1$ and other $w_i$'s as 0, and setting $b\!=\!-(o_1^{(0)}\!+\! o_1^{(1)})/{2}$, we could obtain
\begin{small}
\begin{align}
    \hat{y}^{(0)} & \!:=\! \sigma(\sum_i w_i o_i^{(0)} \!+ \! b)\!=\! \sigma [(o_1^{(0)}\!-\!o_1^{(1)})/2]\! < \! \sigma\lr{0} \!=\! 0.5; \nonumber\\ \nonumber
    \hat{y}^{(1)} &\! :=\! \sigma(\sum_i w_i o_i^{(1)}\! +\! b)\!=\! \sigma[ (o_1^{(1)}\!-\! o_1^{(0)})/2]\! >\!  \sigma\lr{0} \!=\! 0.5.
\end{align}
\end{small}
By taking 0.5 as the decision boundary, we know VSQL could distinguish these two types of input density matrices theoretically. 

Necessity: Assuming all pairs of shadow features $o_i^{(0)}$ and $o_i^{(1)}$ are identical for any group of $\bm{\theta}$, then  $\hat{y}^{(0)} $ and $\hat{y}^{(1)} $ will always be same. Hence, VSQL fails to theoretically distinguish these two input density matrices, which is in contradiction with the condition.
\end{proof}
% Necessity: Assume all pairs of shadow features $o_i^{(0)}$ and $o_i^{(1)}$ are identical for any group of $\bm{\theta}$, then  $\hat{y}^{(0)} $ and $\hat{y}^{(1)} $ will be the same all the time. Hence, VSQL fails to theoretically distinguish these two input density matrices, which is in contradiction with the condition.

\subsection{Proof of Theorem \ref{thm:n_1_n_2_lcoal_shadow}}
\label{appedix:thm:n_1_n_2_lcoal_shadow}

\renewcommand\thethm{\ref{thm:n_1_n_2_lcoal_shadow}}
\begin{thm}
Given two types of $n$-qubit input density matrices $\rho_{in}^{(0)}$ and $\rho_{in}^{(1)}$. If VSQL can not theoretically distinguish them via $m$-local shadow circuits, then neither can via $m^\prime$-local shadow circuits, where $m^\prime<m<n$. And not vice versa.
\end{thm}
\renewcommand{\thethm}{S\arabic{thm}}
\addtocounter{thm}{-1}

\begin{proof}
Sufficiency:
From Corollary \ref{coro:classi_ability_partial_trace}, we know every pair of the corresponding $m$-local partial traces of these two states are identical, i.e.,
\begin{align}
    \lr{\rho_{in}^{(0)}}_{m\text{-local}} = \lr{\rho_{in}^{(1)}}_{m\text{-local}},
\end{align}
where $\lr{\rho}_{m\text{-local}}:=\Tr_{n/{m\text{-local}}}\lr{\rho}$ denotes the partial trace  of $\rho$ on all $n$ other than $m$-local qubit system and the subscripts ``$m$-local'' on both sides mean they are in the same $m$ local qubit system. If we similarly define the following
\begin{align}
    \lr{\rho_{in}^{(0)}}_{m^\prime\text{-local}} & := \Tr_{m/{m^\prime\text{-local}}} \lr{\lr{\rho_{in}^{(0)}}_{m\text{-local}}} \\ \lr{\rho_{in}^{(1)}}_{m^\prime\text{-local}}& := \Tr_{m/{m^\prime\text{-local}}} \lr{\lr{\rho_{in}^{(1)}}_{m\text{-local}}},
\end{align}
then we have 
\begin{align}
     \lr{\rho_{in}^{(0)}}_{m^\prime\text{-local}} = \lr{\rho_{in}^{(1)}}_{m^\prime\text{-local}}.
\end{align}
Due to the arbitrariness of $m$-local, we obtain $m^\prime$-local can also be arbitrary, which means each pair of the corresponding $m^\prime$-local partial traces of these two states are identical. From again Corollary \ref{coro:classi_ability_partial_trace}, we could finish the proof of the Sufficiency part.

A counterexample for Necessity:
Assume 
\begin{align}
    \rho_{in}^{(0)} &=\frac{1}{\sqrt{2}} \lr{\ket{000} +\ket{110}} \cdot \frac{1}{\sqrt{2}} \lr{\bra{000} +\bra{110}} \label{eq:counterexample_0}\\
    \rho_{in}^{(1)} &=\frac{1}{\sqrt{2}} \lr{\ket{000} -\ket{110}} \cdot \frac{1}{\sqrt{2}} \lr{\bra{000} -\bra{110}} \label{eq:counterexample_1}
\end{align}
and let $m^\prime =1$ and $m=2$. In the following, we use $\lr{\rho_{in}}_{i}$ and $\lr{\rho_{in}}_{i,j}$ to denote the 1-local and 2-local partial traces, respectively, where $i, j =1,2,3, i < j$. Now we verify their 1-local and 2-local partial traces:
\begin{align}
    \lr{\rho_{in}^{(0)}}_1 &=\lr{\rho_{in}^{(0)}}_2= \frac{1}{2} \lr{\op{0}{0} +\op{1}{1} }, \quad \lr{\rho_{in}^{(0)}}_3 = \op{0}{0}  \\
 \lr{\rho_{in}^{(1)}}_1 &=\lr{\rho_{in}^{(1)}}_2= \frac{1}{2}\lr{\op{0}{0} +\op{1}{1} }, \quad \lr{\rho_{in}^{(1)}}_3 = \op{0}{0};
\end{align}
\begin{align}
    \lr{\rho_{in}^{(0)}}_{1,2} &=\frac{1}{2} \lr{\ket{00} +\ket{11}} \cdot \lr{\bra{00} +\bra{11}} \\
    \lr{\rho_{in}^{(1)}}_{1,2} &=\frac{1}{2} \lr{\ket{00} -\ket{11}} \cdot \lr{\bra{00} -\bra{11}}.
\end{align}
We see for these two states there exists different 2-local partial traces, even though each pair of their corresponding 1-local partial traces are identical. This indicates, from  Corollary \ref{coro:classi_ability_partial_trace}, VSQL could theoretically distinguish them via 2-lcoal shadow circuits, but could not via 1-lcoal ones. Hence, the two states in Eqs. \eqref{eq:counterexample_0} and \eqref{eq:counterexample_1} could be a successful counterexample for Necessity. This completes the proof.
\end{proof}

\subsection{Proof of Proposition \ref{prop:barren_VSQL}}
\label{appendix:prop:barren_VSQL}

\renewcommand{\thethm}{\ref{prop:barren_VSQL}}
\begin{prop}
If $U_{> l}$ or  $U_{\le l}$ forms at least an $n_{qsc}$-local unitary 2-design,
the mean and the variance of the analytical gradients with respect to $\theta_l$ in VSQL (see Eq. \eqref{eq:ansatz_derivative_theta})  are evaluated as
\begin{align}
    \mathbb E\Lr{\frac{\partial o_i}{\partial_{\theta_l}}} &=0;\qquad\qquad
    \text{Var}\Lr{\frac{\partial o_i}{\partial_{\theta_l}}} = -\frac{1}{4}\cdot \frac{C \lr{\rho_i} }{2^{2n_{qsc}}-1},
\end{align}
where $C\lr{ \rho_i}\in \lr{- 4 \times 2^{n_{qsc}}, 0 }$ denotes a constant  and
$n_{qsc}$ is the number of qubits of the shadow circuits.
\end{prop}
\renewcommand{\thethm}{S\arabic{thm}}
\addtocounter{thm}{-1}

\begin{proof}
Before start, we need the following two lemmas \cite{Dankert2009,Puchaa2017,Cerezo2020}:

\begin{lem} \label{lem:first_moment}
Let $\{U_k\}_{k=1}^K \in \mathcal{U}(d)$ form a unitary $t$-design \cite{Dankert2009} with $t \ge 1$, and let $A,B$ be arbitrary linear operators. Then
\begin{align}\label{eq:first_moment}
\frac{1}{K}\cdot \sum_k \Tr\lr{U_kAU_k^\dag B} = \int_{\mathcal{U}(d)} d\mu_{Haar}(U) \cdot \Tr\lr{UAU^\dag B}= \frac{\Tr(A)\Tr(B)}{d}.
\end{align}
\end{lem}

\begin{lem} \label{lem:second_moment}
 Let $\{U_k\}_{k=1}^K \in \mathcal{U}(d)$ form a unitary $t$-design \cite{Dankert2009} with $t \ge 2$, and let $A,B,C,D$ be arbitrary linear operators. Then
 \begin{align}\label{eq:second_moment1}
\frac{1}{K}\cdot \sum_k \Tr\lr{U_kAU_k^\dag B U_kCU_k^\dag D} =& \int_{\mathcal{U}(d)} d\mu_{Haar}(U) \cdot \Tr\lr{UAU^\dag B UCU^\dag D} \nonumber\\
= &\frac{ \Tr(A)\Tr(C)\Tr(BD)+\Tr(AC)\Tr(B)\Tr(D) }{d^2-1} \nonumber\\
& -\frac{\Tr(AC)\Tr(BD)+\Tr(A)\Tr(B)\Tr(C)\Tr(D) }{d(d^2-1)};
\end{align}
\begin{align}\label{eq:second_moment2}
\frac{1}{K}\cdot \sum_k \Tr\lr{U_kAU_k^\dag B}\Tr\lr{U_kCU_k^\dag D} =& \int_{\mathcal{U}(d)} d\mu_{Haar}(U) \cdot \Tr\lr{UAU^\dag B}\Tr\lr{UCU^\dag D} \nonumber\\
% =\sum_{\alpha,\beta}\int_{\mathcal{U}(d)} d\mu_{Haar}(U) \cdot \Tr\lr{\langle \alpha|UAU^\dag B |\alpha\rangle \langle\beta| UCU^\dag D|\beta\rangle} \\
= &\frac{\Tr(AC)\Tr(BD)+\Tr(A)\Tr(B)\Tr(C)\Tr(D)}{d^2-1} \nonumber\\
& -\frac{\Tr(A)\Tr(C)\Tr(BD)+\Tr(AC)\Tr(B)\Tr(D)}{d(d^2-1)}.
\end{align}
\end{lem}

According to Eq. \eqref{eq:ansatz_derivative_theta}, i.e.,
\begin{align}\label{eq:U_ge_l}
\frac{\partial
    o^{(m)}_i}{\partial \theta_l}& = -\frac{i}{2}\Tr \left(U_{> l}^\dag O U_{>l} \left[P_l, U_{\le l} \rho_i U_{\le l}^\dag\right] \right)   \\
    &= \ \ \  \frac{i}{2}\Tr \left( U_{\le l} \rho_i U_{\le l}^\dag      \left[P_l, U_{> l}^\dag O U_{>l} \right] \right), \label{eq:U_le_l}
\end{align}

(i) if $U_{> l}$ forms at least a $n_{qsc}$-local unitary 2-design, from Eqs. \eqref{eq:U_ge_l}, \eqref{eq:first_moment} and \eqref{eq:second_moment2}, we have
\begin{align}
    \mathbb E\Lr{\frac{\partial o_i^{(m)} }{\partial \theta_l} } &= - \frac{i}{2} \cdot \frac{\Tr\lr{O}\mathbb E \Lr{\Tr\lr{ [P_l, U_{\le l} \rho_i U_{\le l}^\dag]}  }   }{2^{n_{qsc}} } = 0, \\ \label{eq:var_ge_2_design}
    \text{Var}\Lr{\frac{\partial o_i^{(m)} }{\partial \theta_l} } &= - \frac{1}{4} \cdot \frac{\Tr\lr{O^2} \mathbb E \Lr{\Tr\lr{ [P_l, U_{\le l} \rho_i U_{\le l}^\dag]^2 }  }   }{2^{2n_{qsc}} -1 };
\end{align}

(ii) if $U_{\le l}$ forms at least a $n_{qsc}$-local unitary 2-design, from Eqs. \eqref{eq:U_le_l}, \eqref{eq:first_moment} and \eqref{eq:second_moment2}, we have
\begin{align}
    \mathbb E\Lr{\frac{\partial o_i^{(m)} }{\partial \theta_l} } &= \ \ \  \frac{i}{2} \cdot \frac{\Tr\lr{\rho_i} \mathbb E \Lr{ \Tr\lr{ [P_l, U_{> l}^\dag O U_{>l} ]}  }   }{2^{n_{qsc}} } = 0, \\ \label{eq:var_le_2_design}
    \text{Var}\Lr{\frac{\partial o_i^{(m)} }{\partial \theta_l} } &= - \frac{1}{4} \cdot \lr{ \frac{\Tr\lr{\rho_i^2} \mathbb E \Lr{\Tr\lr{ [P_l, U_{> l}^\dag O U_{>l} ]^2 }  }   }{2^{2n_{qsc}} -1 }- \frac{\Tr^2\lr{\rho_i} \mathbb E \Lr{ \Tr\lr{ [P_l, U_{> l}^\dag O U_{>l} ]^2 }  }   }{ 2^{n_{qsc}} (2^{2n_{qsc}} -1) }}.
\end{align}

Now let's consider the term $ \mathbb E \Lr{ \Tr\lr{ [P, U^\dag A U ]^2 }  }$, where $P$ is a Pauli product operator, $U$ denote a series of unitary matrices that the expectation acts on and $A=\sum_j \lambda_j\op{\lambda_j}{\lambda_j}$ denotes a Hermitian operator, where we ssume $\lambda_1 \ge \lambda_2\ge \cdots\ge \lambda_{2^{n_{qsc}}}$. Then we have 
\begin{align}
     \mathbb E \Lr{\Tr\lr{ [P, U^\dag A U ]^2 }}& = \mathbb E \Lr{\Tr (P U^\dag A U - U^\dag A U P )^2 }     \\
    &= 2 \mathbb E \Lr{ \Tr (P U^\dag A U)^2}  -2 \mathbb E \Lr{\Tr( P U^\dag A U U^\dag A U P ) }  \\
    & = 2 \mathbb E \Lr{\sum_{i,j} \lambda_i\lambda_j \Tr (\underbrace{ \bra{\lambda_j} U P U^\dag \op{\lambda_i}{\lambda_i} U P U^\dag \ket{\lambda_j} }_{p_{ij}} ) }  - 2\Tr (A^2) \\
    & = 2 \mathbb E \Lr{ (\vec{\lambda})^\dag P_\Lambda \vec{\lambda }  } -  2\Tr (A^2). \label{eq:def_lambda_P_lambda}
\end{align}
Here, $\vec{\lambda}=[\lambda_1, \lambda_2, \ldots, \lambda_{2^{n_{qsc}}}]^\top $ and we define a matrix $P_\Lambda = \Lr{p_{ij}}$, where each element is defined as
\begin{align}
    p_{ij} = \bra{\lambda_i} U P U^\dag \op{\lambda_j}{\lambda_j} U P U^\dag \ket{\lambda_i}.
\end{align}
From the fact that $p_{ij}\ge 0$ and $\sum_i p_{ij}=\sum_j p_{ij}=1$, we know $P_\Lambda$ is a  doubly stochastic matrix. Now in order to bound the term $(\vec{\lambda})^\dag P_\Lambda \vec{\lambda }$, we can repeatedly perform the following procedure followed from the \textit{Rearrangement inequality}, i.e., for any $i\le k$ and $j\le l$, we have 
\begin{align}\label{eq:rearrangement_ineq}
\begin{matrix}
\lambda_i\lambda_j & + & \lambda_k\lambda_l& \ge & \lambda_i\lambda_l& + & \lambda_k\lambda_j. \\
 p_{ij}  & &  p_{kl} & &  p_{il} & & p_{kj}  \\ 
 \uparrow^{\Delta_1} (or \downarrow_{\Delta_2})  & & \uparrow^{\Delta_1} ( or \downarrow_{\Delta_2}) & &  \downarrow_{\Delta_1} (or \uparrow^{\Delta_2}) & & \downarrow_{\Delta_1} (or \uparrow^{\Delta_2})
\end{matrix}
\end{align}
That is, for the four elements in the four corners of any rectangle (e.g., indexed by rows $i,k$ and columns $j,l$) in $P_\Lambda$, we could increase $p_{ij}$, $p_{kl}$ and decrease $p_{il}$, $p_{kj}$ by $\Delta_1$ simultaneously to get close to its upper bound; Or conversely by $\Delta_2$  to get close to its lower bound (see also Eq. \eqref{eq:rearrangement_ineq}). Here, we can set $\Delta_1=\min\LR{p_{il}, p_{kj}}$ and $\Delta_2=\min\LR{p_{ij}, p_{kl}}$ to satisfy the nonnegativity.
An intuitive example for one step of this procedure is referred to below:
\begin{align}
(\vec{\lambda})^\dagger\begin{bmatrix}
&&j& &l&    \\
  & &\vdots & & \vdots &  \\
i & \cdots  &  0.3\xrightarrow[]{-0.2}0.1  &   \cdots  & 0.4\xrightarrow[]{+0.2}0.6  &   \cdots  \\
  & &\vdots & & \vdots \\
k & \cdots  & 0.5\xrightarrow[]{+0.2}0.7 &   \cdots  & 0.2\xrightarrow[]{-0.2}0.0 &   \cdots  \\
  & &\vdots & & \vdots 
\end{bmatrix} \vec{\lambda} \le 
    (\vec{\lambda})^\dagger
    \overbrace{\begin{bmatrix}
&&j& &l&    \\
  & &\vdots & & \vdots &  \\
i & \cdots  &  0.3 &   \cdots  & 0.4 &   \cdots  \\
  & &\vdots & & \vdots \\
k & \cdots  & 0.5 &   \cdots  & 0.2 &   \cdots  \\
  & &\vdots & & \vdots 
\end{bmatrix}}^{P_\Lambda}
\vec{\lambda}  &  \nonumber \\ 
\le
(\vec{\lambda})^\dagger\begin{bmatrix}
&&j& &l&    \\
  & &\vdots & & \vdots &  \\
i & \cdots  &  0.3\xrightarrow[]{+0.4}0.7  &   \cdots  & 0.4\xrightarrow[]{-0.4}0.0  &   \cdots  \\
  & &\vdots & & \vdots \\
k & \cdots  & 0.5\xrightarrow[]{-0.4}0.1 &   \cdots  & 0.2\xrightarrow[]{+0.4}0.6 &   \cdots  \\
  & &\vdots & & \vdots 
\end{bmatrix}
\vec{\lambda} & .
\end{align}
After a finite number of steps, we will finally obtain
\begin{align}\label{eq:bound_P_lambda}
\sum_{i=1}^{2^{n_{qsc}}} \lambda_i\lambda_{2^{n_{qsc}}-i+1} =
(\vec{\lambda})^\dag 
\begin{bmatrix}
  &  &  & 1\\
  & & 1 & \\
  &  \vdots & & \\
  1&  &  & 
\end{bmatrix}
\vec{\lambda } \le    (\vec{\lambda})^\dag P_\Lambda \vec{\lambda } \le (\vec{\lambda})^\dag 
\begin{bmatrix}
 1 &  &  & \\
  & 1 &  & \\
  &  & \ddots  & \\
  &  &  & 1
\end{bmatrix}
\vec{\lambda } = \sum_{i=1}^{2^{n_{qsc}}} \lambda_i^2.
\end{align}

Substituting Eq. \eqref{eq:bound_P_lambda} into Eq. \eqref{eq:def_lambda_P_lambda}, we have
\begin{align}\label{eq:bound_PUAU}
 2\sum_{i=1}^{2^{n_{qsc}}} \lambda_i\lambda_{2^{n_{qsc}}-i+1}   -  2\Tr (A^2) \le \mathbb E \Lr{\Tr\lr{ [P, U^\dag A U ]^2 }} \le 2 \sum_{i=1}^{2^{n_{qsc}}} \lambda_i^2 -  2\Tr (A^2) =0.
\end{align}

Now we back to prove the variance of the gradients. 

(i) Substituting Eq. \eqref{eq:bound_PUAU} into Eq. \eqref{eq:var_ge_2_design} with $A=\rho_i $, and we define  
$$C(\rho_i) := \Tr\lr{O^2} \mathbb E \Lr{\Tr\lr{ [P_l, U_{\le l} \rho_i U_{\le l}^\dag]^2 }  },$$ 
we have 
\begin{align}
    -4\times 2^{n_{qsc}}    < 2^{n_{qsc}} \lr{0-  2\Tr (\rho_i^2) }      \le  C(\rho_i) \le 0;
\end{align}

(ii) Substituting Eq. \eqref{eq:bound_PUAU} into Eq. \eqref{eq:var_le_2_design} with $A=O=X\otimes\cdots\otimes X $, and we define  
$$C(\rho_i) := \Tr\lr{\rho_i^2} \mathbb E \Lr{\Tr\lr{ [P_l, U_{> l}^\dag O U_{>l} ]^2 }  }  - \frac{\Tr^2\lr{\rho_i} \mathbb E \Lr{ \Tr\lr{ [P_l, U_{> l}^\dag O U_{>l} ]^2 }  }   }{ 2^{n_{qsc}} }.$$ 
Because $O$ has half of the 1 eigenvalues and half of the -1 eigenvalues, we have 
\begin{align}
    -4\times 2^{n_{qsc}}    < \lr{\Tr\lr{\rho_i^2} - \frac{\Tr^2\lr{\rho_i}}{2^{n_{qsc}}}} \lr{  2\sum_{i=1}^{2^{n_{qsc}}}(-1) -2 \Tr(O^2)    }    \le  C(\rho_i) \le 0.
\end{align}

Another point needs to note is that for most of $\theta_l$'s, both $U_{>l}$ and $U_{\le l}$ approximate $n_{qsc}$-local unitary 2-design. Hence, although we give the upper bound 0, most of $C(\rho_i)$ will concentrate to $2\lr{1-2^{n_{qsc}} \Tr(\rho_i^2)}$, which is far from 0 if $\rho_i$ is close to a pure state. This completes the proof.
\end{proof}

\subsection{Proof of Theorem \ref{thm:distinguish_psi_uv}}
\label{appendix:thm:distinguish_psi_uv}

\renewcommand{\thethm}{\ref{thm:distinguish_psi_uv}}
\begin{thm}
Given two families of non-orthogonal 2-qubit quantum states, shown in Eq. \eqref{eq:quantum_data_set}, and each has multiple copies. VSQL could exactly distinguish them, by using only one shadow circuit which consists of only one $R_y$ rotation gate applied on 1-local qubit.
\end{thm}
\renewcommand{\thethm}{S\arabic{thm}}
\addtocounter{thm}{-1}

\begin{proof}
Without loss of generality, we assume $\ket{\psi_u}$ is labelled as `0' and $\ket{\psi_v}$ is labeled as `1'. 
Thus our goal is to show $\hat{y}\lr{\op{\psi_u}{\psi_u};{\theta},\bm{w},b} < 0.5$ and $\hat{y}\lr{\op{\psi_v}{\psi_v};{\theta},\bm{w},b} \ge 0.5$ for any $u, v \in [0, 1]$, i.e.:
\begin{align}
  z_u &:=  w_1 o_1\lr{\op{\psi_u}{\psi_u};{\theta}} + w_2 o_2\lr{\op{\psi_u}{\psi_u};{\theta}} +b  <0   \label{eq:z_u}\\
  z_v &:=  w_1 o_1\lr{\op{\psi_v}{\psi_v};{\theta}} + w_2 o_2\lr{\op{\psi_v}{\psi_v};{\theta}} +b \ge 0  \label{eq:z_v}
\end{align}
could be always satisfied with suitable $w_1, w_2, \theta$ and $b$.

Now we compute these 1-local shadow features from Eq. \eqref{eq:def_observ} as follows:
\begin{align}
    o_1\lr{\op{\psi_x}{\psi_x};{\theta}} &= \bra{\psi_x} \lr{U^\dag(\theta)XU(\theta)\otimes \mathbb I} \ket{\psi_x} \\
    o_2\lr{\op{\psi_x}{\psi_x};{\theta}} & = \bra{\psi_x} \lr{ \mathbb I \otimes U^\dag(\theta)XU(\theta)} \ket{\psi_x},
\end{align}
where $x\in \{u,v\}$  and the shadow circuit $U(\theta)$ is set as $R_y(\theta)$. Since 
\begin{align}
    R_y^\dag(\theta)X R_y(\theta) =\begin{bmatrix}
   \cos\frac{\theta_2}{2} & \sin\frac{\theta_2}{2} \\
    -\sin\frac{\theta_2}{2} & \cos\frac{\theta_2}{2}
    \end{bmatrix} \cdot \begin{bmatrix}
    0 & 1\\
    1 & 0
    \end{bmatrix} \cdot  \begin{bmatrix}
    \cos\frac{\theta_2}{2} & -\sin\frac{\theta_2}{2} \\
    \sin\frac{\theta_2}{2} & \cos\frac{\theta_2}{2}
    \end{bmatrix}  = 
    \begin{bmatrix}
    \sin\theta & \cos\theta \\
    \cos\theta & -\sin\theta 
    \end{bmatrix},
\end{align}
we obtain the observables
\begin{align}
    o_1\lr{\op{\psi_u}{\psi_u};\theta} &=
     \begin{bmatrix}
    \sqrt{1-u^2} & 0 &u & 0
    \end{bmatrix} \cdot  \begin{bmatrix}
    \sin\theta & 0 & \cos\theta & 0 \\
    0&\sin\theta & 0 & \cos\theta  \\
    \cos\theta & 0   & -\sin\theta& 0 \\
    0 & \cos\theta & 0 & -\sin\theta
    \end{bmatrix} \cdot \begin{bmatrix}
   \sqrt{1-u^2}\\
    0\\
    u\\
    0
    \end{bmatrix} \nonumber \\
    &= (1-2u^2)\sin\theta+2u\sqrt{1-u^2}\cos\theta,  \label{eq:o_1_u} \\
    o_2\lr{\op{\psi_u}{\psi_u};\theta} &=
     \begin{bmatrix}
    \sqrt{1-u^2} & 0 &u & 0
    \end{bmatrix} \cdot  \begin{bmatrix}
    \sin\theta &  \cos\theta & 0 & 0 \\
    \cos\theta &-\sin\theta & 0 & 0  \\
    0 & 0 & \sin\theta & \cos\theta \\
    0 & 0 & \cos\theta  & -\sin\theta
    \end{bmatrix} \cdot \begin{bmatrix}
   \sqrt{1-u^2}\\
    0\\
    u\\
    0
    \end{bmatrix} \nonumber \\
    &= \sin\theta; \label{eq:o_2_u}
\end{align}
And similarly
\begin{align}
    o_1\lr{\op{\psi_v}{\psi_v};\theta}=(1-2v^2)\sin\theta, \label{eq:o_1_v}\\
    o_2\lr{\op{\psi_v}{\psi_v};\theta}=(2v^2-1)\sin\theta.  \label{eq:o_2_v}
\end{align}
Substituting Eqs. \eqref{eq:o_1_u}, \eqref{eq:o_2_u}, \eqref{eq:o_1_v} and \eqref{eq:o_2_v} into Eqs. \eqref{eq:z_u} and \eqref{eq:z_v}, we have
\begin{align}
    z_u &=w_1\Lr{ (1-2u^2)\sin\theta+2u\sqrt{1-u^2}\cos\theta} + w_2 \sin\theta +b, \label{eq:z_u_compute} \\
    z_v & = \lr{w_1 -w_2} (1-2v^2)\sin\theta  +b. \label{eq:z_v_compute}
\end{align}

As  $w_1, w_2, \theta$ and $b$ are chosen arbitrarily and $u,v\in[0,1]$, without loss of generality, we assume $0<\sin\theta\le \cos\theta<1$, then 
\begin{align}\label{eq:bound_o_1_u}
     (1-2u^2)\sin\theta+2u\sqrt{1-u^2}\cos\theta \ge  \lr{1-2u^2+2u\sqrt{1-u^2}}\sin\theta \ge -\sin\theta.
\end{align}
If we set $w_2 < w_1 <0$, combining with Eq. \eqref{eq:bound_o_1_u}, we have
\begin{align}
    z_u &\le (w_2-w_1)\sin\theta +b, \\
    z_v &\ge (w_2- w_1)\sin\theta +b,
\end{align}
where both the equal signs (`=') occur only if $u=v=1$.
Therefore, if we want $z_u<0$ and $z_v\ge 0$ all the time, it's sufficient to have the following conditions:
\begin{align}
\left\{\begin{aligned}
   & 0<\sin\theta\le \cos\theta<1,\\
   & w_2 < w_1 <0, \\
   & (w_2-w_1)\sin\theta +b =0. 
   \end{aligned}
   \right.
\end{align}
Of course we could also have other settings for $\theta, w_1, w_2, b$ that satisfy our requirements, but here, one is enough. This completes the proof.
\end{proof}

% \end{linenumbers}

\end{document}